\documentclass[11pt]{article}

\usepackage[hidelinks]{hyperref}
\usepackage{graphicx} %
\usepackage{verbatim}
\usepackage{amsmath,amssymb,amsthm}
\usepackage[roman,full]{complexity}
\usepackage{tabularx,booktabs}
\usepackage{fullpage}
\usepackage{pifont}
\usepackage{xstring}
\usepackage[inline]{enumitem}
\usepackage{tcolorbox}
\usepackage{algorithm,algpseudocode}
\usepackage{ifthen}

\usepackage{multicol,multirow,makecell}

\usepackage{tikz}
\usetikzlibrary{positioning}

\newcommand{\cald}{\mathcal{D}}
\newcommand{\cale}{\mathcal{E}}

\newcommand{\cali}{\mathcal{I}}

\newcommand{\calr}{\mathcal{R}}
\newcommand{\cals}{\mathcal{S}}

\newcommand{\calu}{\mathcal{U}}

\newcommand{\dash}{\hbox{-}\allowbreak}

\newcommand{\LANGDEF}[3]{
\begin{center}
{\small 
\begin{tabularx}{0.98\columnwidth}{ll}
\toprule
\multicolumn{2}{c}{\textsc{#1}} \\
\midrule
{\bf Given:}   & \parbox[t]{0.75\columnwidth}{#2\vspace*{1mm}}  \\
{\bf Question:}& \parbox[t]{0.75\columnwidth}{#3\vspace*{.5mm}} \\ 
\bottomrule
\end{tabularx}
}
\end{center}
}

\newcommand{\PROBDEF}[3]{
\begin{center}
{\small 
\begin{tabularx}{0.98\columnwidth}{ll}
\toprule
\multicolumn{2}{c}{\textsc{#1}} \\
\midrule
{\bf Given:}   & \parbox[t]{0.75\columnwidth}{#2\vspace*{1mm}}  \\
{\bf Compute:}& \parbox[t]{0.75\columnwidth}{#3\vspace*{.5mm}} \\ 
\bottomrule
\end{tabularx}
}
\end{center}
}

\newtheorem{theorem}{Theorem}[section]

\newtheorem{lemma}[theorem]{Lemma}
\newtheorem{proposition}[theorem]{Proposition}

\newtheorem{observation}[theorem]{Observation}
\newtheorem{construction}[theorem]{Construction}

\newtheorem{example}[theorem]{Example}

\newtheorem{definition}[theorem]{Definition}

\newcommand{\naturals}{\ensuremath{\mathbb{N}}}
\newcommand{\posnat}{\ensuremath{\mathbb{N}^+}}

\newcommand{\rmmath}[1]{\ensuremath{{\rm #1}}}

\newcommand{\card}[1]{\lVert#1\rVert}

\newcommand{\score}{\ensuremath{{\rm score}}}

\newcommand{\approval}{\ensuremath{{\rm approval}}}
\newcommand{\dd}{\hbox{-}}
\newcommand{\crc}[1]{%
\ifthenelse
{\equal {#1} {}}
{\rmmath{CRC}_{\naturals}}
{\rmmath{CRC}_{\leq #1}}%
}

\title{The Cost of Failure: On The Complexity of Recampaigning under Fixed Districts}
\author{Michael C. Chavrimootoo and Aidan Jeansonne\\
Department of Computer Science\\
Denison University\\
Granville, OH 43023}

\date{August 20, 2026}

\usepackage[autostyle=false, style=english]{csquotes}
\MakeOuterQuote{"}

\newif\ifhideproofs

\ifhideproofs
\usepackage{environ}
\NewEnviron{hide}{}

\fi

\begin{document}

\maketitle

\begin{abstract}%
    Redistricting efforts have gathered contemporary attention in both 
    popular
    and scholarly debates, particularly in the United States where efforts to redraw congressional districts to favor either of the two major parties in 12 states---such as California, Texas, and Ohio---have captured the public eye. 
    The treatment of redistricting in computational social choice has essentially focused on the process of determining "appropriate" districts.
    In this work, we are interested in understanding the gamut of options left for the "losing" party, and so we consider the flip side of the problem: Given fixed/predetermined districts, can a given party still make their candidates win by strategically placing them in certain districts?
    We dub this as "recampaigning" to capture the intuition that a party would redirect their campaigning efforts from one district to another. 
    We model recampaigning as a computational problem, consider natural variations of the model, and study those new models through the lens of (1) (polynomial-time many-one) interreducibilities, (2) separations/collapses (both unconditional and axiomatic-sufficient), and (3) both worst-case and parametrized complexity.
\end{abstract}

\section{Introduction}

Across the United States of America (USA), at least 12 states are pursuing measures to redraw their district lines in a way that favors one party over all others, and these efforts have been successful in some states, but not all of them~\cite{okr-coh-rig-sch:news:tracking-redistricting}. This process is more commonly known as gerrymandering or redistricting and has been well-studied in recent years.

From the perspective of computational social choice, redistricting is itself a type of electoral control, which is a type of electoral attack in which a chair (or organizer) of an election has the ability to alter the structure of the election to influence its outcome. The common actions are to add/delete/partition candidates/voters. Work on control was initiated by Bartholdi, Tovey, and Trick~\cite{bar-tov-tri:j:control} and further extended by Hemaspaandra, Hemaspaandra, and Rothe~\cite{hem-hem-rot:j:destructive-control}. For a comprehensive overview on control, see the chapters on the topic in \cite{fal-rot:b:handbook-comsoc-control-and-bribery,rot:b:econ-second-edition}. On a very simple level, partitioning of voters can be viewed as a way of modeling the action of drawing two districts, and there have been extensions of that model to account for multiple districts~\cite{erd-hem-hem:c:more-natural-models-of-partition-control,fal-gou-lan-les-mon:c:president-problem}.

Since then, additional work has aimed at studying redistricting by considering the geographical placements of voters or ``connectivity'' using graphs~\cite{lew-lev-ros:c:geographic-manipulation,coh-lew-ros:c:gerrymandering-graphs,eib-fom-pan-sim:c:districts-fine-grained,gup-jai-pan-roy-sau:c:gerrymandering-on-graphs-complexity}. The notion of redistricting has also been extensively studied along other dimensions, and we refer readers to the comprehensive book by Duchin and Walch~\cite{duc-wal:b:political-geometry} on the subject.
In more recent years, the notion of ``reverse gerrymandering'' has received attention. In that model, districts are preformed and voters have the opportunity to move across districts~\cite{lev-lew:c:reverse-gerrymandering,pal:j:priced-gerrymandering}. This captures the very real phenomenon where people move to regions where they feel their vote is more likely to have an impact~\cite{hen:news:voters-moving}.

In this paper, we consider a different side of this problem. In particular, we are guided by the following question:

\begin{quote}
If a party is unable to redraw district lines in their favor (for whichever variety of reasons may hold), can they still elect their designated candidates across the various districts, without changing the placement of voters?
\end{quote}

Relating back to the aforementioned examples in the USA, our work captures an underlooked consequence of redistricting, which we argue is just as important to study as redistricting. For example, consider a state with two dominant political parties $A$ and $B$, where $A$ controls the state's legislative and has redrawn the districts in its favor. Then $B$'s attempts to redraw district lines would be futile, and they should instead focus on electing winners across the existing districts. This movement of candidates across parties has been observed several times in recent USA history. For example, in 2024 House member Lauren Boebert ran in a different district in the same state to be reelected to the House of Representatives~\cite{bed:news:boebert}.

But this is not an issue unique to the USA\@. Indeed, in the United Kingdom, Nigel Farage has run in 8 different constituencies since 1994 and was only recently (2024) elected to parliament~\cite{reg-plc-sai:news:farage}. In India, Rahul Gandhi switched constituencies from Amethi to Wayanad to maintain his seat in parliament~\cite{tim:news:gandhi}. Another interesting case is that of Kristina Keneally in Australia who switched electoral division in 2022 and faced defeat~\cite{cal-dan:news:australia-keneally}. This last case highlights that in the real world, gauging whether one's odds are better in a different division/district can in fact be nontrivial.

To approach the above guiding question, we introduce a new model of candidate control that involves placing candidates strategically in predetermined districts so as to make them winners. 
This can be viewed as the chair of a given party trying to make a set of their candidates win across various districts. From an optimization perspective, this can also be viewed as the problem of maximizing the number of winners from that given party. 
This type of modeling differs from ``apportionment'' in that we are not concerned with dividing a given number of seats, but rather, we are concerned with determining suitable candidates without exceeding the number of seats available. We provide additional examples in Section~\ref{s:prelims} to support that view.
To our knowledge, this is the first type of work that treats the problem of simply assigning candidates to districts to model the fact that redrawing the district lines is possibly hard or maybe even impossible (thereby making incumbent on the party to determine where its candidates should campaign to win). %

\textbf{Contributions:} We introduce the notion of recampaigning and provide a model to study it from a computational perspective. 
We consider several variations to our model, and contrast the variations in several ways. We consider the potential interreducibilities between our variants and find that in general, those do not exist (except for two trivial cases). We also compare our variants using the notion of ``separations and collapses''~\cite{car-cha-hem-nar-tal-wel:j:sct} thereby highlighting that our models are indeed distinct. Finally, we study the worst-case and parametrized complexities of our various problems under several natural voting rules, such as the $t$-approval and $t$-veto families of voting rules, and Condorcet; we also give more general results that hold for all voting rules satisfying certain properties.

\textbf{Organization:} Section~\ref{s:prelims} defines notation and terminology used in this paper, along with our models. Section~\ref{s:relationships} explores relationships between variants of the model, Section~\ref{s:worst-case} contains our worst-case complexity results, and Section~\ref{s:parametrized} contains our parametrized-complexity results. 
The Conclusion section discusses future work.
\section{Preliminaries}\label{s:prelims}

Let $\naturals = \{0, 1, 2, \ldots\}$ and let $\posnat = \{1, 2, 3, \ldots\}$. 
We use the shorthand for each $n \in \posnat$ that $[n] = \{1, \ldots, n\}$ and $[n]_0 = \{0\} \cup [n]$.
In this paper, every graph $G=(V, E)$ is undirected, and we represent an edge in such a graph as a pair $(a, b)$. 
We also assume standard familiarity with worst-case complexity (e.g., \NP-completeness) and parametrized complexity (e.g., fixed-parameter tractability and associated notions of hardness). 
As is standard in complexity theory, for each string $x$, we denote the length of $x$ by $|x|$ and for each set $A$, we denote the cardinality of $A$ by $\card{A}$.
For a set $A$, we let $2^A$ denote the power set of~$A$. %

\subsection{Elections and Voting Rules}

An election is a pair $(C, V)$, where $C$ is a finite set of candidates and $V$ is a finite collection of votes, where each vote is 
a linear order over $C$\@.

A voting rule $\cale$ is a function that maps an election $(C, V)$ to a (possibly empty) subset of $C$, i.e., the winner set. We often abuse notation and let $\cale$ refer both to the function and to the ``name'' of the voting rule.
Moreover, for each voting rule $\cale$, we define the language $W_\cale = \{(C, V, p) \mid p \in \cale(C, V) \}$, which we at times refer to as the winner problem for $\cale$.

When describing votes, we use the following shorthands. If a set is included in a vote, e.g., if $C = \{a, b, c, d\}$ and $D = \{b,c\}$, the vote $D > a > d$ refers to any vote where the candidates in $D$ are ordered in some predetermined polynomial-time computable fashion, e.g., in lexicographical order of candidate name. To refer specifically to the lexicographical ordering of candidates in a vote, we add a right arrow on top, e.g., $\overset{\rightarrow}{D} > a > d$ . If we instead write $\overset{\leftarrow}{D}$, then we mean the reverse lexicographical ordering of the candidates. In a vote, the shorthand ``$\cdots$'' refers to all other candidates not yet specified in the vote.

Several of the results we give will be about so-called (polynomial-time uniform) pure scoring rules~\cite{bet-dor:j:possible-winner-dichotomy}. To understand what those are, we must first speak of scoring vectors.
A scoring vector for an $m$-candidate election $(C, V)$ is a list $(\alpha_1, \ldots, \alpha_m)$ of nonincreasing integers, which effectively determines how many points each candidate receives from a given vote based on their position in that vote.
A pure scoring rule $\calr$ consists of an infinite family of scoring vectors $s_1, s_2, \ldots$ such that for each $i \in \posnat$, $s_i$ is a scoring vector for an $i$-candidate election, $s_i$ is computable in polynomial time with respect to $i$, and $s_{i+1}$ can be obtained from $s_i$ solely by inserting a new value into the list $s_i$. 
The winner set of an election $(C, V)$ under pure scoring rule $\calr$ is computed as follows. First compute the scoring vector $s_{\card C} = (\alpha_1, \ldots, \alpha_{\card C})$ and for each vote $v \in V$, give a candidate at position $i \in [\card C]$---where the top-ranked candidate is at position~1 and last-ranked candidate is at position ${\card C}$---$\alpha_i$ points. Candidates with maximal score are winners.

We say that a scoring rule is nontrivial if for any $m > 1$, the scoring vector for $m$ candidates contains at least two distinct values.  
Let $t \geq 1$. The pure scoring rule $t$-approval gives a candidate one point for each vote in which they are ranked among the first $t$ ranked candidates. The pure scoring rule $t$-veto gives every candidate one point for each vote in which they are \emph{not} among the \emph{last} $t$ ranked candidates. 
\subsection{Our Models}\label{sub:models}
In this paper, we study the computational problem of ``recampaigning'' wherein simultaneous elections across $k$ districts are taking place under a voting rule $\cale$ and an entity is putting forward a set of candidates $A$ with the goal of making every such candidate a winner in some district (with possibly multiple winners in a district). As explained in our Introduction, this is a natural consideration to take and quite relevant, especially in the United States, where parties may seek to maximize the number of their candidates that are winners across the districts in a state, subject to the constraint that they are unable to change the district lines.

For each voting rule $\cale$ and each positive integer $\ell$, we define the following problem.

\LANGDEF{$\cale$ Constructive $\ell$-Recampaigning ($\cale\dash\crc{\ell}$)}
{A set $A$ of $n$ additional candidates and $k$ districts $D_1, \ldots, D_k$, where for each $i \in [k]$, $D_i = (C_i, V_i)$ is an $\cale$ election, $C_i \cap A = \emptyset$ and votes in $V_i$ are over $C_i \cup A$.}
{Is there an assignment of candidates in $A$ to districts such that every $a \in A$ is a winner in the district they are assigned to and each district electing a candidate in $A$ has at most $\ell$ winners?}

 The ``Constructive'' in the above problem's name refers to the fact that the goal of this problem is to make candidate(s) win. 
 Moreover, there is no restriction on how many candidates can be assigned to each district. However, it is worth noting that in $\ell$-recampaigning, a successful assignment of candidates to districts will never be assigning more than $\ell$ candidates to a given district.
 
 In the above problem, the case where $\ell = 1$ is a natural case to consider if each district elects exactly one winner. However, that may not always be the case; many elections around the world run elections where multiple candidates may be elected from one district. One such country is Mauritius, where each constituency (which, in our terminology, corresponds to a  district) elects three candidates to elect the members of parliament~\cite{wik:url:mauritius-elections}, using what is essentially 3-approval. 
 In a different context, award committees often select multiple recipients subject to some threshold, for example by selecting up to, but no more than, three recipients.
 In companies, if the Board of Directors has vacant seats (e.g., following resignations), then an election may be conducted to fill as many seats as possible. However, if not every candidate is deemed suitable, some of those seats may remain unfilled.
 It is thus natural to consider the case where there may be multiple winners in each district, but that number is bounded. Moreover, the reason why we bound the number of winners only for the case where a candidate in $A$ is a winner is because the entity determining the assignment of candidates to districts should not be concerned with the happenings of a district in which it has no participating candidate; we believe that doing so would be unnatural and our current framing is less restrictive.

Another interesting variant that we also consider in this paper is the case where there is no bound on the number of winners in each district. One may view this as analogous to the standard cowinner/nonunique-winner model in the study of manipulative attacks on elections.
Indeed, the notion of bounding the number of allowed winners in the resulting election has been commonly studied under the terms ``unique-winner model'' and ``nonunique-winner model.'' In our terminology, the unique-winner model is analogous to 1-recampaigning and the nonunique-winner model is analogous to unbounded recampaigning. However, our work views the two standard winner models like a spectrum to provide more refined analyses.
To our knowledge, there are no natural voting rules under which a control problem's complexity depends on the winner model under consideration. However, as we will show, the complexity of recampaigning can depend on whether the number of winners is bounded or unbounded.
We define that problem below for each voting rule $\cale$.

\LANGDEF{$\cale$ Constructive Unbounded Recampaigning ($\cale\dash\crc{}$)}
{A set $A$ of $n$ additional candidates and $k$ districts $D_1, \ldots, D_k$, where for each $i \in [k]$, $D_i = (C_i, V_i)$ is an $\cale$ election, $C_i \cap A = \emptyset$ and votes in $V_i$ are over $C_i \cup A$.}
{Is there an assignment of candidates in $A$ to districts such that every $a \in A$ is a winner in the district they are assigned to?}

For each of the above-defined problems, we consider the priced version which includes as part of the input a price function $\pi : [k] \times A \to \posnat$---where $\pi(i, a)$ is the cost to add a candidate $a \in A$ to district $D_i$---and a budget $B \geq 0$, and the priced problems all have the additional constraint that an assignment of candidates to districts cannot cost more than $B$. For each voting rule $\cale$ and $\ell \geq 1$ respectively denote those priced problems as $\cale\dash\$\crc{\ell}$ and $\cale\dash\$\crc{}$. 

As we shall see, the complexity of $\ell$-recampaigning (which we also call ``bounded recampaigning'') is not necessarily tied to the complexity of the underlying voting rule. However, the complexity of unbounded recampaigning is never lower than that of the underlying voting rule.
\begin{proposition}
    For each voting rule $\cale$, $W_\cale \leq_m^p \cale\dd\crc{}$.%
\end{proposition}
\begin{proof}
    Let $\cale$ be a voting rule and let $(C, V, p)$ be an instance of $W_\cale$. Our polynomial-time reduction constructs the set of additional candidates $A = \{p\}$, and the sole district is $D_1 = (C-A, V)$.
    $(A, D_1) \in \cale\dd\crc{}$ if and only if  $p \in \cale(C, V)$, which by the definition of $W_\cale$ is equivalent to saying that $(A, D_1) \in \cale\dd\crc{}$ if and only if $(C, V, p) \in W_\cale$. Thus the reduction is correct.
\end{proof}

\section{Relationships and Properties}\label{s:relationships}

As we introduce new models for electoral control, it is important to recognize relationships between such models. Understanding those relationships forces us to have a deeper understanding of the models under treatment, but also lessens the burden on researchers. For example, existing reductions may help in providing classification results. On the other hand, viewing the relationships between the languages defined by those problems may be useful in pruning the number of models under consideration; we point to the work of Carleton et al.~\cite{car-cha-hem-nar-tal-wel:j:sct} (and later extended by~\cite{car-cha-hem-nar-tal-wel:j:search-versus-search,cha-cli-fer-gib-hem-luu-nar-wan:c:axiomatic-tools}) who found many previously unknown relationships (in terms of set equality and containment for the associated languages) between control problems that had been studied actively for years.

\begin{proposition}\label{prop:containments}
    Let $\cale$ be a voting rule and let $\ell$ and  $\ell'$ be positive integers, such that $\ell < \ell'$. Then
    \begin{enumerate}%
        \item $\cale\dd\crc{\ell} \subseteq \cale\dd\crc{\ell'} \subseteq \cale\dd\crc{}$ and
        \item $\cale\dd\$\crc{\ell} \subseteq \cale\dd\$\crc{\ell'} \subseteq \cale\dd\$\crc{}$.
    \end{enumerate}
    Moreover, there is a natural voting rule under which the above containments are strict. 
\end{proposition}
\begin{proof}
    Let $\cale$ be a voting rule and let $\ell$ and $\ell'$ be two positive integers such that $\ell < \ell'$.
    Let $\cali \in \cale\dd\crc{\ell}$. Then by definition, each district where an additional candidate wins has at most $\ell < \ell'$ winners, so $\cali$ is in $\cale\dd\crc{\ell'}$ and also $\cale\dd\crc{}$ since the latter problem places no bounds on the number of winners. The same argument establishes the corresponding relationships in the priced setting.
    Let $\cale'$ be a voting rule such that for any finite set of candidates $C$, $\cale'(C, \emptyset) = C$. Many natural voting rules, such as $t$-approval for all $t \geq 1$, satisfy that condition.

    For each $n \geq 1$, let $\cali_n = (A_n, D_1)$, where $A_n = \{a_1,\ldots, a_n\}$ and $D_1 = (C_1, V_1)$ is an empty election, i.e., both $C_1$ and $V_1$ are empty collections. Then it follows that $I_{\ell'} \in \cale'\dd\crc{\ell'}$ but not in $\cale'\dd\crc{\ell}$, and $I_{\ell'+1} \in \cale'\dd\crc{}$, but not in $\cale'\dd\crc{\ell'}$. For the corresponding results in the priced setting, set all prices to 1 and let the budget be $\card{A_n}$, for each $n \geq 1$.
\end{proof}

A natural property of voting rules that 
permeates both pure social choice theory and computational social choice
is the notion of resoluteness, i.e., the property that a voting rule does not elect more than one candidate as winner on any input.\footnote{Some authors may require voting rules to always output one winner, in which case resoluteness ensures \emph{exactly} one winner when the candidate set is nonempty.} It is a natural notion to consider as often one wishes to elect at  most one winner. Examples of resolute voting rules include Condorcet~\cite{con:b:condorcet-paradox} and Ranked Pairs~\cite{tid:j:clones-dodgson,hem-lav-men:j:schulze-and-ranked-pairs}. However, it is not unusual to consider voting rules that elect at most $k$ winners, especially in the context of multiwinner voting rules (see~\cite{rot:b:econ-second-edition} for a comprehensive overview). We generalize that definition and prove related results as they concern our new models.

\begin{definition}[\textsc{$k$-Resolute}]
    For each integer $k \geq 1$ and voting rule $\cale$,
    $\cale$ is said to satisfy the {\rm \textsc{$k$-Resolute}} condition if $\cale$ never elects more than $k$ candidates as winners.
\end{definition}

\begin{proposition}\label{prop:k-resolute}
    For each voting rule $\cale$ and positive integer $L$, if $\cale$ is \textsc{$L$-Resolute}, then for each $\ell \geq L$
    $\cale\dd\crc{\ell} = \cale\dd\crc{}$
    and $\cale\dd\$\crc{\ell} = \cale\dd\$\crc{}$.
\end{proposition}
\begin{proof}%
    Let $L$ be a positive integer, let $\cale$ be a \textsc{$L$-Resolute} voting rule, and let $\ell \geq L$.
    The $\subseteq$ relationship follows from Proposition~\ref{prop:containments}.
    Let $\cali' \in \cale\dd\crc{}$.
     Because $\cale$ is \textsc{$L$-Resolute}, the number of winners in a district is at most $\ell$, so unbounded recampaigning simplifies to $\ell$-recampaigning. Therefore using the same assignment of additional candidates to districts it follows that $\cali' \in \cale\dd\crc{\ell}$.
     The same argument also establishes 
     that $\cale\dd\$\crc{\ell} = \cale\dd\$\crc{}$.
\end{proof}

We now switch gears and focus on interreducibilities, in terms of polynomial-time many-one reductions.

\begin{observation}\label{observation:trivial-relationships}
    For each voting rule $\cale$ and $\ell \geq 1$,
    \begin{enumerate*}[label=]
        \item $\cale\dd\crc{\ell} \leq_m^p \cale\dd\$\crc{\ell}$ and
        \item $\cale\dash\crc{} \leq_m^p \cale\dash\$\crc{}$.
    \end{enumerate*}
\end{observation}

    The above reductions are in fact logspace computable, but for our purposes, polynomial-time computability suffices.

While Observation~\ref{observation:trivial-relationships} tells us that there are trivial reductions from an unpriced problem to its corresponding priced version, the next result shows that the other direction does not necessarily hold. The crux of the result is to encode information about hard problems into the price function.
We also show there is no general interreducibility between bounded and unbounded recampaigning.
The proof of the following result leverages results proven later in this paper, so we encourage readers to skip the next proof until after reading Sections~\ref{s:worst-case} and~\ref{s:parametrized}.

\begin{proposition}\label{prop:priced-unpriced}
The following statements are equivalent:
\begin{enumerate}
    \item $\P \neq \NP$.
    \item There is a polynomial-time computable voting rule such that 
    priced 
    recampaigning does not reduce to unpriced 
    recampaigning.
    \item There is a polynomial-time computable voting rule such that bounded recampaigning does not reduce to unbounded recampaigning.
    \item There is a polynomial-time computable voting rule such that unbounded recampaigning does not reduce to bounded recampaigning.
\end{enumerate}
\end{proposition}
\begin{proof}
    Since all of our recampaigning problems are in $\NP$ when the underlying voting rule is polynomial time computable, it follows that (2), (3), and (4) imply (1). 
    It remains to show that (1) implies~(2), (3), and (4).

    \noindent{\boldmath$(1) \implies (2)$}.
    
    Let $\cale_1$ be a voting rule that on input $(C, V)$ outputs $C$ if $\card C = 3$, and outputs $\emptyset$ otherwise. 
    
    We give a simple algorithm to show that $\cale_1\dd\crc{3} \in \P$.
    Given an input $(A, D_1, \ldots, D_k)$, 
    for each $i \in [3]$ let $N_i$ denote the number of districts with exactly $3-i$ candidates, i.e., the subscript $i$ denotes the ``slack'' or the ``number of candidates that must be added'' in that district.
    Then $(A, D_1, \ldots, D_k)$ is a Yes instance of $\cale_1\dd\crc{3}$ if and only if there exists nonnegative integer values $\alpha$, $\beta$, $\gamma$ such that $\alpha \leq N_1$, $\beta \leq N_2$, $\gamma \leq N_3$, and $\alpha + 2\beta + 3\gamma = \card{A}$.
    Since each $N_i$ is at most $k$, we can check all possible assignments of values to $\alpha$, $\beta$, and $\gamma$ in $k^3$ iterations. Thus $\cale_1\dd\crc{3} \in \P$.

    On the other hand, we show that $\cale_1\dd\$\crc{3}$ is \NP-complete by giving a reduction for the Exact Cover by 3-Sets (X3C) problem, a standard \NP-complete problem~\cite{kar:b:reducibilities}.

    \LANGDEF{Exact Cover by 3-Sets (X3C)}{A nonempty universe of elements $\calu = \{u_1, \ldots, u_{3m}\}$ and a nonempty collection $\cals = \{S_1, \ldots, S_k\}$ of 3-element subsets of $\calu$.}{Is there a set $\cals' \subseteq \cals$ of size $m$ such that every element of $\calu$ is contained in some element of $\cals'$?}

    The reduction proceeds as follows. Given an arbitrary instance $(\calu, \cals)$ of X3C, we create for each $S_i$ a district $D_i$ with no voters and no candidates. There are thus $k$ districts. We also let $A = \calu$. We then set the price function $\pi$ as follows. For each $i \in [k]$ and $a \in A$, set $\pi(i, a) = 1$ if $a \in S_i$ and set $\pi(i, a) = 3m+1$ otherwise. Finally, set the budget $B = 3m$.

    Consider an assignment of candidates $A_i$ to district $D_i$. Then $\cale_1(C_i \cup A_i, V_i)$ elects at most three candidates and elects every candidate in $A_i$ if and only if $A_i = S_i$. In other words, a successful assignment of candidates to districts so that every additional candidate is a winner corresponds exactly to $m$ 3-sets of additional candidates that correspond uniquely to elements to $\cals'$, i.e., an exact cover of $\calu$ by 3-sets in $\cals$.

    Moreover, notice that $\cale_1$ satisfies the 3-\textsc{Resolute} condition, so $\cale_1\dd\crc{3} = \cale_1\dd\crc{}$ and $\cale_1\dd\$\crc{3} = \cale_1\dd\$\crc{}$ by Proposition~\ref{prop:k-resolute}. Thus priced bounded/unbounded recampaigning does not reduce to its unpriced version.

    \noindent{\boldmath$(1) \implies (3)$}.

    Define $\cale_2$ to be a voting rule whose vote type is linear orders over candidates that takes as input $(C, V)$ and outputs $C$ if $\card{C} \geq 4$, and otherwise outputs the 1-approval winner(s) of $(C, V)$. Using the same construction as in the proof of Theorem~\ref{theorem:3-and-above}, $\cale_2\dd\crc{3}$ and $\cale_2\dd\$\crc{3}$ are both $\NP$-complete. 

        On the other hand, $\cale_2\dd\crc{}$ is in $\P$. Indeed, consider an arbitrary instance of $(A, D_1, \ldots, D_k)$. If $\card{A} \geq 4$, then place all additional candidates in any district, as all candidates will win. If $\card{A} \leq 3$, then there are at most $k^3$ assignments of additional candidates to districts, which can all be checked in polynomial time.\footnote{Though we do not give a proof, using the same approach as in the proof of Proposition~\ref{prop:priced-unpriced} yields the NP-completeness of $\cale_2\dd\$\crc{}$. So this construction does not show the lack of a reduction from unpriced/unpriced bounded recampaigning to priced unbounded recampaigning. However, that is beyond the scope of our proposition's claim. It would nonetheless be an interesting case to explore.}

    \noindent{\boldmath$(1) \implies (4)$}.
    
    This follows directly from Theorems~\ref{theorem:fpt-districts} and~\ref{theorem:para-np-hard}, 
    and the fact that $\FPT = {\rm para}\dd\NP$-hard if and only $\P=\NP$~\cite{flu-gro:b:parameterized-complexity}.
    \end{proof}

We next consider another family of conditions on voting rules, which we will appeal to in order to characterize the complexity of 1-recampaigning in Section~\ref{s:worst-case}.

\begin{definition}
    For each $k \geq 1$, 
    a voting rule $\cale$ is said to satisfy the $k$-{\rm \textsc{Winner}} condition if there is a polynomial-time algorithm that given an $\cale$ election $(C, V)$ and $p\in C$, accepts if $p \in \cale(C, V)$ and $\card{\cale(C, V)} \leq k$, and rejects otherwise.
\end{definition}

The $k$-\textsc{Winner} condition may seem to imply something strong about the computability of $\cale$ in order to satisfy $k$-\textsc{Winner}, such as polynomial-time computability. Naturally if $\cale$ is polynomial-time computable, then for each $k\geq 0$, it satisfies $k$-\textsc{Winner}. However, the converse is not true. Indeed, the only guarantee is that winner determination is polynomial-time computable on \emph{some} inputs. As the next result demonstrates, one can design voting rules that are hard to compute and yet satisfy $k$-\textsc{Winner}.

\begin{proposition}
    For each $k \geq 1$, 
    there is an uncomputable voting rule $\cale$ that satisfies $k$-{\rm \textsc{Winner}}.
\end{proposition}
\begin{proof}

    Let $L$ be an undecidable subset of $\naturals$.  We define the voting rule $\cale$ as follows

        \begin{align*}
        \cale(C, V) = 
        \begin{cases}
            C & \text{ if } \card{C} \leq k \text{ or } \card{C} \in L\\
            \\
            \emptyset & \text{ otherwise}.
        \end{cases}
    \end{align*}

    It is easy to see verify $\cale$ is not computable, and yet it satisfies $k$-Winner. A candidate is one of $k$ winners if and only if there are at most $k$ candidates in the election, which can be checked in polynomial time.
    Moreover, if $\cale$ were computable, then 
    $L$ would be decidable,
    which is a contradiction.
\end{proof}

Finally, we connect the $k$-\textsc{Winner} condition to recampaigning as follows.

\begin{lemma}\label{lemma:p-to-k-winner}
    For each voting rule $\cale$ and $\ell \geq 1$, if $\cale\dd\crc{\ell} \in \P$, then $\cale$ satisfies $\ell$-{\rm\textsc{Winner}}.
\end{lemma}
\begin{proof}
    Let $\cale$ be a voting rule,  let $\ell \geq 1$, and
    suppose $\cale\dd\crc{\ell}$ is in~$\P$. We give below a polynomial-time algorithm that checks if a given candidate in a $\cale$ election is one of $k$ winners.
    
    Let $(C, V)$ be a $\cale$ election and let $p \in C$ be inputs to our polynomial-time algorithm. We construct the following $\cale\dd\crc{\ell}$ instance. There is one district $D_1 = (C-\{p\}, V)$ and the set of additional candidates is $A = \{p\}$. Thus $(D_1, A) \in \cale\dd\crc{\ell}$ if and only if $p \in \cale(C, V)$ and $\card{\cale(C, V)} \leq \ell$, which can be checked in polynomial-time.
\end{proof}

\section{Worst-Case Complexity}\label{s:worst-case}

We first focus on the worst-case complexities of our problems with respect to natural voting rules. In particular, we completely characterize the complexity of 1-recampaigning for all voting rules, we completely determine the complexity of 2-recampaigning for pure scoring rules, and for all other types of recampaigning, we establish hardness results for the $t$-approval and $t$-veto families of scoring rules. We also establish easiness results for certain voting rules.

\subsection{1-Recampaigning}

In this subsection, we characterize the conditions on voting rules for which 1-recampaigning is in P\@. To do so, we reframe our problem in terms of the Minimum-Cost Unbalanced Assignment problem, which has a polynomial-time algorithm~\cite{ram-tar:t:unbalanced-assignment}.

\begin{lemma}\label{lemma:simpleuw-to-p}
    For each voting rule $\cale$ that satisfies the {\rm \textsc{1-Winner}} condition, both 
    $\cale\dd\crc{1}$ and $\cale\dd\$\crc{1}$ are in~$\P$.
\end{lemma}
\begin{proof}
    Let $\cale$ be a voting rule that satisfies the {\rm \textsc{1-Winner}} condition. By Observation~\ref{observation:trivial-relationships}, it suffices to give a polynomial-time algorithm for the priced 1-recampaigning problem.

    In the algorithm we present, we use as a subroutine a polynomial-time algorithm for the minimum-cost unbalanced assignment problem~\cite{ram-tar:t:unbalanced-assignment}.

    \PROBDEF{Minimum-Cost Unbalanced Assignment~\cite{ram-tar:t:unbalanced-assignment}}
    {A bipartite graph $G=(L \cup R, E)$, where $\card L \geq \card R$, and a weight function $w: E \to \naturals$.}
    {A maximum cardinality matching $M$ of $G$ whose weight/cost $w(M)=\sum_{e \in M} w(e)$ is minimized.}
    
    We prove our result by showing that the task of assigning candidates in districts so as to make them unique winners is essentially the same task as performing a minimum-cost unbalanced assignment. Indeed, because the problem at hand requires each additional candidate to be a winner and each district can have at most one winner, it follows that a successful assignment of additional candidates never places two (or more) additional candidates in the same district.
    Thus the problem reduces to finding a subset of districts that can be one-to-one mapped to additional candidates under the restriction that each additional candidate is a unique winner in their district. Moreover, because the assignment is budget-constrained, this corresponds to the ``minimum-cost'' nature of the unbalanced assignment problem.

    Let the input to our algorithm for $\cale\dd\$\crc{1}$ be
    $(D_1, \ldots, D_k, A, \pi, B)$.
    Let $\cald = \{D_1, \ldots, D_k\}$.
    Our algorithm constructs the undirected bipartite graph $G = (A \cup \cald, E)$, where for each $i\in [k]$ and $c \in A$,
    \begin{center}
    $(c, D_i) \in E$ if and only if $\{c\} = \cale(C_i \cup\{c\}, V_i)$,
    \end{center}
    and for each $(c, D_i) \in E$, let $w(c, D_i) = \pi(i, c)$.
    Since $\cale$ satisfies the {\rm \textsc{1-Winner}} condition, it takes polynomial time to check if a pair $(c, D_i)$ should be in $E$, and thus $G$ is polynomial-time computable. Our algorithm then computes a minimum-cost maximal matching $M$ of $G$.
    If $\card{M} = \card{A}$ and $w(M) \leq B$, then the algorithm accepts. Otherwise, the algorithm rejects.

    This works because there is a solution to the $\cale\dd\$\crc{1}$ problem on input $(\cald, A)$ if and only if $G$ has a maximal matching of size $\card{A}$ with weight at most $B$.
    Thus our algorithm runs in polynomial time and is correct. %
\end{proof}

In the following result, for each voting rule $\cale$, we denote by $\P^{\cale}$ the class of languages that are accepted by deterministic polynomial-time machines with a function oracle (i.e., a blackbox function) to compute $\cale$ in one step.

\begin{proposition}
    For each voting rule $\cale$, both 
    $\cale\dd\crc{1}$ and $\cale\dd\$\crc{1}$ are in $\P^{\cale}$.
\end{proposition}
\begin{proof}[Proof (sketch)]
    Notice that in the proof of Lemma~\ref{lemma:simpleuw-to-p}, the only task that depends on the complexity of computing the function $\cale$ is the task of computing $E$, which can be computed in polynomial time relative to the (function) oracle $\cale$.
\end{proof}

Taking stock of our results, we give the main theorem of this subsection.

\begin{theorem}\label{theorem:crc-1}
    Let $\cale$ be a voting rule. The following statements are equivalent.
    \begin{enumerate}
        \item $\cale$ satisfies the {\rm \textsc{1-Winner}} condition.
        \item $\cale\dd\$\crc{1} \in \P$.
        \item $\cale\dd\crc{1} \in \P$.
    \end{enumerate}
\end{theorem}
\begin{proof}
    Lemma~\ref{lemma:simpleuw-to-p} proves that (1) implies (2), Observation~\ref{observation:trivial-relationships} shows that (2) implies (3), and finally Lemma~\ref{lemma:p-to-k-winner} proves that (3) implies (1).
\end{proof}

\subsection{Unbounded and \boldmath$\ell$-Recampaigning, for $\ell > 1$}

We now consider the case where the bound on the number of winners is at least two. We completely determine the complexity of  priced/unpriced 2-recampaigning for all polynomial-time uniform scoring protocols. We show that for all voting rules in the $t$-approval and $t$-veto families of voting rules and all $\ell > 2$, $\ell$-recampaigning and unbounded recampaigning are $\NP$-complete. And we also show, using minimum-weight perfect $b$-matchings for multigraphs (defined later), that bounded/unbounded priced/unpriced recampaigning is in polynomial time for trivial scoring rules.

\begin{theorem}\label{theorem:trivial-scoring-rule}
    Let $\calr$ be a trivial scoring rule, i.e., a voting rule that elects all candidates as winners. Then for each $\ell \geq 2$, $\ell$-recampaigning under $\calr$ is in $\P$, in the priced and unpriced settings. This also holds in the unbounded recampaigning setting.
\end{theorem}
\begin{proof}

    Let $\ell \geq 2$. We give the algorithm for priced $\ell$-recampaigning, which will immediately imply the truth of the result for unpriced $\ell$-recampaigning. We note in passing that the reader may already see why the unpriced setting has a polynomial-time algorithm: a simple greedy algorithm that allocates additional candidates in each district so to have no more than $\ell$ candidates in each district solves the problem. In the priced setting however, the greedy approach no longer works, and it is easy to devise an example demonstrating that.

    To prove our result for the priced setting, we will appeal to the fact that Min-Weight Perfect $b$-Matching for Multigraphs is in $\P$~\cite{fit-hem:c:weighted-matching}.

    \LANGDEF{Min-Weight Perfect $b$-Matching for Multigraphs~\cite{fit-hem:c:weighted-matching}}
    {A weighted multigraph $G=(V, E)$ with weight function $w: E \to \naturals$, a capacity function $b: V \to \naturals$, and an integer $k \geq 0$.}
    {Does $G$ have a perfect matching of weight at most $k$, i.e., does there exist a set of edges
$M \subseteq E$ of weight at most $k$ such that each vertex $v \in G$ is incident to exactly $b(v)$ edges in $M$?}

    Let $D_1, \ldots, D_k$ be $k$ districts, let $A$ be the set of $n$ additional candidates, let $\pi$ be the price function for the problem, and let $B$ be the budget.

    For each district $D_i$, define $\Delta_i = \max(0, \ell-\card{C_i})$, i.e., the maximum number of additional candidates that can be added to that district.

    We then define the following instance of Min-Weight Perfect $b$-Matching for Multigraphs.
    \begin{itemize}
        \item $G = (V, E)$, where
        \begin{itemize}
        \item $V = A \cup \{d_1, \ldots, d_k\} \cup \{s\}$,
        \item for each $a \in A$ and $i \in [k]$, $E$ contains one edge $(a, d_i)$ with weight $\pi(i, a)$,
        \item for each $i \in [k]$, $E$ contains $\Delta_i$ times the edge $(d_i, s)$ with weight 0,
        \end{itemize}
        \item for each $a \in A$, $b(a) = 1$,
        \item for each $i \in [k]$, $b(d_i) = \Delta_i$,
        \item $b(s) = \max\left(0, \left(\sum_{i=1}^k \Delta_i\right) - \card{A}\right)$,
        \item set the weight limit $k$ to $B$.
    \end{itemize}

    \begin{example}
        Suppose $\ell = 3$.
        Let $A = \{a_1, a_2, a_3\}$, let $D_1$ and have no candidates, let $D_2$ have one candidate, and let $D_3$ have four candidates. Moreover, let $\pi(1, a_1) = 5$, 
        $\pi(1, a_2) = 8$, 
        $\pi(1, a_3) = 10$, 
        $\pi(2, a_1) = 3$, 
        $\pi(2, a_2) = 1$, 
        $\pi(2, a_3) = 16$, 
        $\pi(3, a_1) = \pi(3, a_2) = \pi(3, a_3) = 0$, and 
        $B=16$. Then the constructed graph is the one below, and each solution is captured by the min-weight perfect $b$-matching of the graph below. 
        \begin{center}
        \small
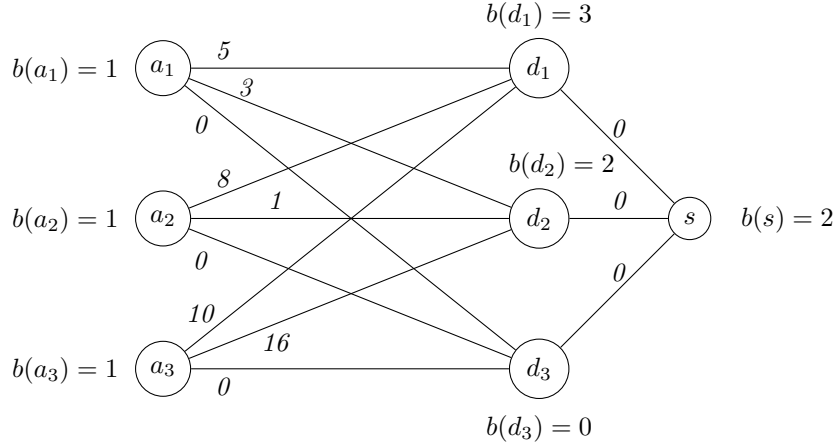
        \begin{tikzpicture}[node distance = 2cm, node/.style={minimum width = 2cm}]
            \node[draw, circle, label={[xshift=-13mm, yshift=-7mm]$b(a_1) = 1$}] (a1) {$a_1$};
            \node[draw, circle, below of = a1, label={[xshift=-13mm, yshift=-7mm]$b(a_2) = 1$}] (a2) {$a_2$};
            \node[draw, circle, below of = a2, label={[xshift=-13mm, yshift=-7mm]$b(a_3) = 1$}] (a3) {$a_3$};
            
            \node[draw, circle, right of = a1, xshift=3cm, label={$b(d_1) = 3$}] (d1) {$d_1$};
            \node[draw, circle, below of = d1, label={[xshift=3mm]$b(d_2) = 2$}] (d2) {$d_2$};
            \node[draw, circle, below of = d2, label={[yshift=-1.5cm]$b(d_3) = 0$}] (d3) {$d_3$};

            \node[draw, circle, right of = d2, label={[xshift=13mm, yshift=-6mm]$b(s) = 2$}] (s) {$s$};

            \draw (a1) -- node[above, xshift=-1.7cm] {5} ++ (d1) ;
            \draw (a1) -- node[above, xshift=-1.4cm, yshift=5mm] {3} ++ (d2) ;
            \draw (a1) -- node[above, xshift=-20mm, yshift=10mm] {0} ++ (d3) ;

            \draw (a2) -- node[above, xshift=-17mm, yshift=-7mm] {8} ++ (d1) ;
            \draw (a2) -- node[above, xshift=-1cm] {1} ++ (d2) ;
            \draw (a2) -- node[above, xshift=-20mm, yshift=2mm] {0} ++ (d3) ;

            \draw (a3) -- node[above, xshift=-20mm, yshift=-15mm] {10} ++ (d1) ;
            \draw (a3) -- node[below, xshift=-10mm, yshift=-4mm] {16} ++ (d2) ;
            \draw (a3) -- node[below, xshift=-1.7cm] {0} ++ (d3) ;

            \draw (s) -- node[above] {0} ++ (d1) ;
            \draw (s) -- node[above] {0} ++ (d2) ;
            \draw (s) -- node[above] {0} ++ (d3) ;
        \end{tikzpicture}
        \end{center}

        Here the optimal solution is to place $a_1$ in $D_2$, $a_2$ in $D_2$, and $a_3$ in $D_1$, which has cost 14 and corresponds to the minimum-weight perfect $b$-matching of the above graph.
    \end{example}

    We now prove the correctness of our algorithm by proving that  there is a successful assignment of additional candidates to districts without exceeding the budget $B$ if and only if $G$ has perfect $b$-matching with weight at most $k$.
    
    $\implies$:
    Suppose there is a successful assignment of additional candidates to districts without exceeding the budget $B$, and let $A_i$ denote the additional candidates assigned to $D_i$, for $i \in [k]$. We construct the following perfect $b$-matching $M$ for $G$. For each $i \in [k]$ and $a \in A_i$, add edge $(a, d_i)$ to $M$. Moreover for each $i\in [k]$, add $\Delta_i - \card{A_i}$ copies of $(d_i, s)$ to $M$. 
    Note that by the properties of $\ell$-recampaigning, it follows that $\card{A_i} \leq \Delta_i$, so the difference in the previous sentence is never negative.
    Moreover, note that those edges in $M$ incident to $s$ contribute nothing to the weight of $M$, so

    \begin{align*}
        \sum_{e \in M} w(e) = \sum_{i=1}^k \sum_{a \in A_i} w(a, d_i) = \sum_{i=1}^k \sum_{a \in A_i} \pi(i, a) \leq B.
    \end{align*}

    Moreover, since every additional candidate is assigned to exactly one district, there is exactly one edge incident to each $a \in A$, i.e., $b(a) = 1$. 
    Let $i \in [k]$.
    There are $\card{A_i}$ edges in $M$ incident to $d_i$ but not $s$. We also have $\Delta_i - \card{A_i}$ edges in $M$ between $d_i$ and $s$. So $d_i$ is incident to $\Delta_i$ edges in $M$, i.e., $b(d_i) = \Delta_i$. Finally, notice that the number of edges incident to $s$ in $M$ is

    \begin{align*}
        \sum_{i=1}^{k} \left(\Delta_i - \card{A_i}\right) = \left(\sum_{i=1}^{k} \Delta_i\right)  - \left(\sum_{i=1}^{k}  \card{A_i}\right)  = \left(\sum_{i=1}^k \Delta_i\right) - \card{A} = b(s).
    \end{align*}

    Therefore $M$ is a perfect $b$-matching with weight at most $k$.

    $\impliedby$:
    Let $M$ be a perfect $b$-matching with weight at most $k$.

    For each $i \in [k]$, let $A_i = \{c \mid (c, d_i) \in M\}$. Since the vertices in $A$ only have edges to vertices in $\{d_1, \ldots, d_k\}$, this corresponds to an assignment to candidates to districts, and because the matching is perfect, each candidate is assigned to exactly one district (and naturally, multiple candidates can be assigned to the same district).
    It remains to show that the assignment defined by $A_1, \ldots, A_n$ is a valid solution to the priced $\ell$-recampaigning problem with budget $B = k$.

    Let $i \in [k]$ and assume $A_i \neq \emptyset$ (otherwise, there's no candidate assigned to district $D_i$).
    Recall that $b(d_i) = \Delta_i$, so $\card{A_i} \leq \Delta_i = \ell - \card{C_i}$. Thus the number of winners in $(C_i \cup A_i, V)$ is at most $\ell$. Thus every candidate electing an additional candidate elects no more than $\ell$ candidates. Finally, the weight of the assignment is 
    \begin{align*}
        \sum_{i=1}^{k} \sum_{a\in A_i} \pi(i, a) = \sum_{i=1}^{k} \sum_{a\in A_i} w(a, d_i) = w(M) \leq B.
    \end{align*}

    The above construction can be done in polynomial-time as $\card{V} = n+ k +1$ and $\card{E} \leq nk + k\ell$,
    and $\P$ is closed downwards under polynomial-time many-one reductions, thus concluding this part of the proof.

    Finally we argue that unbounded recampaigning under $\calr$ is also in $\P$. Notice that in the above construction, $\ell \leq n$, so by modifying the construction to let $\ell =n$, the size of the resulting graph is still polynomial in the size of the input, and the correctness argument still holds. Thus the modified reduction is a correct polynomial-time many-one reduction to Min-Weight Perfect $b$-Matching for Multigraphs, and unbounded recampaigning under $\calr$ is in $\P$.
\end{proof}

To contrast our previous result, we find that under 1-approval, unpriced 2-recampaigning is $\NP$-complete. 
To prove that result, we will first show that Exactly-3-3DM (defined below) is $\NP$-complete and then provide a reduction from that problem. We note that Fitzsimmons and Hemaspaandra~\cite{fit-hem:c:conformant-elections} also use that problem to prove their \NP-completeness results, but they do not provide a proof that Exactly-3-3DM is \NP-complete, so we include one for completeness.

    \LANGDEF{Exactly-3 3-Dimensional Matching (Exactly-3-3DM)}
    {Three pairwise-disjoint sets $W$, $X$, and $Y$, each having $k$ elements, a set $S \subseteq W \times X \times Y$, where each element in $W \cup X \cup Y$ appears in exactly three triples in $S$.}
    {Is there an $\hat{S} \subseteq S$ of size $k$ so that no two elements of $\hat{S}$ agree on any coordinate?}

\begin{proposition} 
    {\rm Exactly-3-3DM} is \NP-complete.
\end{proposition}
\begin{proof} The NP membership is straightforward, so we will focus on the NP-hardness by reducing from R3DM, which is $\NP$-complete~\cite{gar-joh:b:int}.

    \LANGDEF{Restricted 3-Dimensional Matching (R3DM)~\cite{gar-joh:b:int}}
    {Three pairwise-disjoint sets $W$, $X$, and $Y$, each having $k$ elements, a set $S \subseteq W \times X \times Y$, where each element in $W \cup X \cup Y$ appears in at most three triples in $S$.}
    {Is there an $\hat{S} \subseteq S$ of size $k$ so that no two elements of $\hat{S}$ agree on any coordinate?}

    Let $I=(W,X,Y,S)$ be an instance of R3DM where $k=\card{W}=\card{X}=\card{Y}$. If $\card{S}=3k$, then each element in $W \cup X\cup Y$ appears in exactly 3 triples in $S$, and this is an instance of Exactly-3-3DM, so we make no changes and let $I' = I$. Otherwise, 
    let $W'=W$, $X'=X$, $Y'=Y$, and $S'=S$, and proceed as follows.

    Find $w \in W', x \in X', y \in Y'$ such that $w$, $x$, and $y$ are each in less than three triples in $S'$. (Since each triple in $S'$ contains one element from each of $W'$, $X'$, and $Y'$, there must be some selection of $w$, $x$, and $y$ satisfying this condition.) Let $i = \card{S'}$. Create unique elements $w^{i}, x^{i}, y^{i}$ and add these elements to $W'$, $X'$, and $Y'$, respectively. We then add the following triples to $S'$: $(w, x^i,y^i)$, 
    $(w^i,x,y^i)$, $(w^i,x^i,y)$, and $(w^i,x^i,y^i)$. Each of $w$, $x$, and $y$ now appear in one more triple (in $S'$) than before, and each of $w^i$, $x^i$, and $y^i$ appears in exactly three triples in $S'$.
    The process described in this paragraph is repeated until $\card{S'} = 3\card{X'}$ (it's easy to see that this process terminates in at most $2k$ steps).
    Once the above process terminates, let $I' = (W', X', Y', S')$, which is
    an instance of Exactly-3-3DM. We now show that $I \in {\rm R3DM}$ if and only if $I' \in \text{Exactly-3-3DM}$.

    $\implies$:
    Suppose $I\in {\rm R3DM}$, and let $\hat{S}$ be a 3-dimensional matching. 
    All triples in $S$ remain in $S'$. Therefore $\hat{S} \subseteq S'$. For each newly created elements $w^i$, $x^i$, and $y^i$, $(w^i,x^i,y^i)\in S'$. So, there exists  a set $\hat{S'} \supseteq \hat{S}$ that also contains each $(w^i,x^i,y^i)$
    with all necessary constraints satisfied. Therefore, $I' \in \text{Exactly-3-3DM}$.

    $\impliedby$:
    Suppose $I' \in \text{Exactly-3-3DM}$, and let $\hat{S'}$ be a 3-dimensional matching. For each $w^i$, $x^i$, and $y^i$, $(w^i,y^i,x^i) \in \hat{S'}$. If this triple were not chosen, either $w^i$ or $x^i$ or $y^i$ would not be present in a triple of $\hat{S'}$. The remaining elements to be chosen are exactly those in $W \cup X \cup Y$. These must be selected from the subset of triples in $S'$ which do not contain any newly created element, i.e., $S$. Thus there must exist a 3-dimensional matching $\hat{S} \subseteq S$, so $I \in$ R3DM. 
\end{proof}

\begin{theorem}\label{theorem:1-approval-l2}
    Both $1\dd\approval\dd\$\crc{2}$ and $1\dd\approval\dd\crc{2}$ are \NP-complete.  
\end{theorem}
\begin{proof}

Clearly, $1\dd\approval\dd\crc{2} \leq_m^p 1\dd\approval\dd\$\crc{2}$ and both problems are in $\NP$, so in this proof we only give the NP-hardness of $1\dd\approval\dd\crc{2}$.

Consider an instance $(W, X, Y, S)$ of Exactly-3-3DM and construct an instance of $1\dd\approval\dash\crc{2}$ as follows. Let $A=W\cup X$ and for each $y \in Y$, we will define a district $D_y = (C_y, V_y)$, where $C_y = \emptyset$ and $V_y$ will be defined later in this proof.
For each $y \in Y$, 
let $S_y=\{(w_1,x_1,y),(w_2,x_2,y), (w_3,x_3,y)\}$ contain the three elements of $S'$ containing $y$.
Finally, let $A' = A - \{w_1, w_2, w_3, x_1, x_2, x_3\}$.
We now define for each $y \in Y$ what $V_y$ is, based on the structure of $S_y$. For each case below, the bulleted list is precisely the list of votes in $V_y$.\footnote{Note that we could simplify this construction here to prove our result. However, this slightly more complex exposition will be useful when proving Theorem~\ref{theorem:nontrivial-2}.}

\noindent\textbf{Case~1: \boldmath$w_1 = w_2 = w_3$ and \boldmath$\card{\{x_1, x_2, x_3\}} = 3$}.
\begin{itemize}
        \item $w_1 > x_1>x_2 > x_3 > \overset{\rightarrow}{A'}$,
        \item $x_1>x_2>x_3>w_1 > \overset{\rightarrow}{A'}$
\end{itemize}

\noindent\textbf{Case~2: \boldmath$w_1 = w_2 \neq w_3$ and $x_1 = x_3 \neq x_2$.}
\begin{itemize}
        \item $w_3>x_2>w_1>x_1 > \overset{\rightarrow}{A'}$
        \item $w_3>w_1>x_2>x_1 > \overset{\rightarrow}{A'}$
        \item $x_1>w_1>w_3>x_2 > \overset{\rightarrow}{A'}$
        \item $x_2>x_1>w_3>w_1 > \overset{\rightarrow}{A'}$
\end{itemize}

\noindent\textbf{Case~3: \boldmath$w_1 = w_2 \neq w_3$ and $\card{\{x_1, x_2, x_3\}} = 3$.}
\begin{itemize}
        \item $w_3>w_1>x_3>x_1>x_2>\overset{\rightarrow}{A'}$
        \item $x_3>w_1>w_3>x_2>x_1>\overset{\rightarrow}{A'}$
        \item $w_3>x_1>x_2>x_3>w_1>\overset{\rightarrow}{A'}$
        \item $x_3>x_1>x_2>w_3>w_1>\overset{\rightarrow}{A'}$
\end{itemize}

\noindent\textbf{Case~4: \boldmath$\card{\{w_1, w_2, w_3\}} = 3$ and $\card{\{x_1, x_2, x_3\}} = 3$}. 
\begin{itemize}
        \item $w_1 > w_2 > w_3 > x_3 > x_2 > x_1 > \overset{\rightarrow}{A'}$
        \item $w_1 > x_2 > x_3 > w_3 > w_2 > x_1 > \overset{\rightarrow}{A'}$
        \item $x_1 > w_2 > w_3 > x_3 > x_2 > w_1 > \overset{\rightarrow}{A'}$
        \item $x_1 > x_2 > x_3 > w_3 > w_2 > w_1 > \overset{\rightarrow}{A'}$
        \item $w_2 > w_1 > w_3 > x_3 > x_2 > x_1 > \overset{\rightarrow}{A'}$
        \item $w_2 > x_1 > x_3 > w_3 > x_2 > w_1 > \overset{\rightarrow}{A'}$
        \item $x_2 > w_1 > w_3 > x_3 > w_2 > x_1 > \overset{\rightarrow}{A'}$
        \item $x_2 > x_1 > x_3 > w_3 > w_2 > w_1 > \overset{\rightarrow}{A'}$
        \item $w_3 > w_1 > w_2 > x_3 > x_2 > x_1 > \overset{\rightarrow}{A'}$
        \item $w_3 > x_1 > x_2 > x_3 > w_2 > w_1 > \overset{\rightarrow}{A'}$
        \item $x_3 > w_1 > w_2 > w_3 > x_2 > x_1 > \overset{\rightarrow}{A'}$
        \item $x_3 > x_1 > x_2 > w_3 > w_2 > w_1 > \overset{\rightarrow}{A'}$
\end{itemize}

The remaining cases can be analogously obtained from the above cases (e.g., if $x_1 = x_2 = x_3$, then the resulting construction is analogous to Case~1).

\begin{observation}\label{observation:add-two-cands}
Since $\card{W}=\card{X}=\card{Y}$, we have $\card{Y}$ districts, and no more than two candidates can be added to a district, we must add exactly two candidates from $A$ to each district. 
\end{observation}

For each $y \in Y$, let $n_y = \card{V_y}$ (which is even). We prove the following lemma first.

\begin{lemma}\label{lemma:valid-cands}
For each $y \in Y$ and 
for each  $a,b \in A$, 
the winner set of $(\{a,b\}, V_y)$ under 1-approval is $\{a,b\}$ if and only if $(a,b,y) \in S$.
\end{lemma}
\begin{proof} 
Let $y \in Y$ and let $a,b \in A$.
First notice that $(a,b,y) \in S$ if and only if $(a,b,y) \in S_y$.

$\impliedby$:
Assume $(a, b, y) \in S_y$.
In 1-approval election $(\{a,b\},V_y)$, both $a$ and $b$ receive a score of $n_y/2$ and so both are elected as winners.

$\implies:$
Assume $(a,b,y) \not\in S_y$. We consider the following cases.

\noindent\textbf{Case~A: \boldmath$a$ and $b$ are elements of two distinct triples in $S_y$.}

This represents the case where $a$ or $b$ could each be a cowinner in district $y$, but not when paired together. In such a case, in 1-approval election $(\{a,b\},V_y)$, $a$ receives a score of $n_y/2-1$ and $b$ receives a score of $n_y/2+1$ or vice versa. So either $a$ or $b$ is a winner, but not both. 

\noindent\textbf{Case~B: \boldmath$a$ is an element of a triple in $S_y$, but $b$ is not, or vice versa.}

Suppose $b$ is not present in any of the triples of $S_y$, but $a$ is. Then $a$ will always be ranked higher than $b$, and so $a$ receives $n_y$ points and $b$ receives 0. Thus $a$ is elected as the sole winner.

\noindent\textbf{Case~C: \boldmath Neither $a$ nor $b$ are elements of a triple in $S_y$. }

Suppose $a$ appears before $b$ in lexicographic ordering. Then, $a$ is ranked above $b$ in each vote in $V_y$. Therefore, $a$ receives $n_y$ points and $b$ receives 0. $a$ is the only winner under this election.
\end{proof}
$\implies$:
Suppose $I=(W,X,Y,S)$ is a Yes-instance of Exactly-3-3DM with certificate $\hat{S}\subseteq S$. For each $y\in Y$, there exists $(w,x,y) \in \hat{S}$, where $w,x \in A$, so we can add both $w$ and $x$ to district $D_y$. By Lemma~\ref{lemma:valid-cands}, $w$ and $x$ are the only winners of district $D_y$. After repeating this for all $y\in Y$, all candidates in $A$ have been added to a district where they win.
So, $(A, D_1, \ldots, D_{\card Y})$ is a Yes-instance of $1\dd\approval\dd\crc{2}$.

$\impliedby$:
Suppose $(A, D_1, \ldots, D_{\card Y})$ is a Yes-instance of $1\dd\approval\dd\crc{2}$. 
By Observation~\ref{observation:add-two-cands}, we can only add exactly two candidates from $A$ to each district. Each candidate can only be added to one district, and all candidates in $A$ have been added. For each $D_y$, let $w_y$ and $x_y$ be the candidates added from $A$ which are winners of the resulting election in that district. From Lemma~\ref{lemma:valid-cands}, $(w_y,x_y,y) \in S$. 
Let $\hat{S} = \{(w_y, x_y, y) \mid y \in Y\}$. It follows that $\hat{S}$ is 3-dimensional matching and so it is a Yes-instance of Exactly-3-3DM\@.
\end{proof}

By modifying the above proof, we are able to \emph{fully determine} the complexity of 2-recampaigning for nontrivial scoring rules.

\begin{theorem}\label{theorem:nontrivial-2}
    For each nontrivial pure scoring rule $\calr$,
    both $\calr\dd\$\crc{2}$ and $\calr\dd\crc{2}$ are
    $\NP$-complete.
\end{theorem}
\begin{proof}
The winner problem for any pure scoring rule is in \P, so both problems are in NP. In light of Observation~\ref{observation:trivial-relationships}, we will only show the NP-hardness of $\calr\dd\crc{2}$, which we do by modifying the proof of Theorem~\ref{theorem:1-approval-l2}.

\begin{observation}\label{proposition:nontriv-m}
    For any nontrivial pure scoring rule $\calr$, we can find in constant time some $m'$ such that for all $m \ge m'$, the scoring vector $s_m$ is nontrivial. 
\end{observation}
We first find a nontrivial scoring vector $s_m$ for $\calr$, where $m > 1$. Because the scoring vector consists of nonincreasing values, there is an $i \in [m-1]$, such that $\alpha_1=\cdots=\alpha_i>\alpha_{i+1}$. 
Given an instance $(W, X, Y, S)$ of Exactly-3-3DM, we first perform the reduction used in the proof of Theorem~\ref{theorem:1-approval-l2}, and modify the construction further below.

For each district $D_y$, $y \in Y$, let $C_y=\{s_1,\ldots,s_{i-1},s_{i+2},\ldots,s_m\}$. From Observation~\ref{observation:add-two-cands}, this ensures that the election in each district contains $m$ candidates. 

Let $V_y' = V_y$. We will redefine $V_y$ later to effectively contain $i$ copies of $V_y'$. 
For each $k \in [i-1]_0$, define the collection of votes $V_{y, k}$ as follows
\begin{itemize}
    \item if $ k = 0$, then for each $v \in V_y'$, $V_{y, 0}$ contains vote $s_1 > \cdots > s_{i-1} > v > s_{i+2} > \cdots > s_m$,\footnote{The ``$v$'' in the vote should be interpreted to mean that we list the candidates listed in $v$ in the same order they appear in $v$.} and
    \item if $k > 0$, then for each $v \in V_{y}'$, $V_{y, k}$ contains vote $v > s_1 > \cdots > s_{k-1} > s_{k+1} > \cdots > s_m > s_k$.
\end{itemize}

Finally, let $V_y = \bigcup_{k \in [i-1]_0} V_{y, k}$, and prove the analogous result to Lemma~\ref{lemma:valid-cands}.

\begin{lemma}
    For each $y \in Y$ and 
for each $a,b \in A$, 
the winner set of $\calr(C_y \cup \{a,b\}, V_y) = \{a,b\}$  if and only if $(a,b,y) \in S$.
\end{lemma}
\begin{proof}
    Let $y \in Y$ and let $a,b \in A$.
    Let $p=\alpha_1=\cdots=\alpha_i$ be the score given to the first $i$ ranked candidates. Let $q=\alpha_{i+1}\ge\cdots\ge\alpha_m$ be the score given to the next candidate. Let $n_y=\card{V_{y,0}}$. We will consider the same cases as in Lemma~$\ref{lemma:valid-cands}$. In this proof, let $\score(x)$ denote the score of candidate $x$ using votes from $V_y$.

    $\impliedby:$    
    Let $(a, b, y) \in S$.
    In $V_{y,0}$, the first $i-1$ scoring values are given to  candidates in $C_y$. The next two scoring values are $p,q$ where $p>q$. Since $a$ and $b$ are ranked in the $i$th place in half of the votes and in the $(i+1)$th place in the other half, their scores in $V_{y,0}$ are each $p(n_y/2)+q(n_y/2)$. 
    
    For each of the remaining $V_{y,1},\ldots,V_{y,i-1}$, $a,b$ are the top two ranked candidates. If $i=1$, the first two scoring values are not equal and no copies of the votes were made. Thus $\calr(C_y \cup \{a,b\}, V_y) = \{a,b\}$. Suppose now that $i > 1$.
    So for each $k \in [i-1]$, $a$ and $b$ receive $p(n_y)$ points from $V_{y, k}$.
    Therefore the scores of $a$ and $b$ in $V_y$ are $$p(n_y/2)+q(n_y/2)+p(n_y)(i-1).$$ To verify they are winners, we need to verify that no other candidates scored more than this.
    We first will look at candidates $s_1,\ldots,s_{i-1}$. In $V_{y,0}$, each of these candidates receives $p(n_y)$ points from being ranked in the top $i$ positions for each vote. For each remaining $V_{y,1},\ldots,V_{y,i-1}$, it is not immediately clear the scores they will receive. However, by construction we know that in $V_{y,k}$, $s_k$ is ranked last and receives $\alpha_m$, and $\alpha_m \le q$. For the remaining copies, their score is at most $p$. So for each $k \in [i-1]$, when combining the scores of all votes in district $D_y$, $$\score(s_k)\le p(n_y)+p(n_y)(i-2)+q(n_y)=p(n_y)(i-1)+q(n_y).$$ And, since $p>q$, $$\score(s_k) \leq p(n_y)(i-1)+q(n_y) < p(n_y)(i-1)+p(n_y/2)+q(n_y/2) = \score(a) = \score(b).$$

    Last, we consider candidates $s_{i+2}, \ldots,s_m$. These candidates score $\alpha_{i+2},\ldots,\alpha_m$ in $V_{y,0}$ and $\alpha_{i+1},\ldots,\alpha_{m-1}$ in the remaining votes. Therefore they always receive at most $q$ points for each vote, and their score is at most $q (n_y)  i$ which is less than the scores of $a$ and $b$. Therefore, 
    $\calr(C_y \cup \{a,b\}, V_y) = \{a,b\}$.

    $\implies:$ Let $(a,b,y) \not\in S$.
    For this case, we will once again show that $a$ and $b$ will have different scores and cannot both be winners. Note that in each of $V_{y,1},\ldots,V_{y,i-1}$, $a$ and $b$ receive the same score. It suffices to show they receive different scores in $V_{y,0}$.
    
    \noindent\textbf{Case~A: \boldmath$a$ and $b$ are elements of two distinct triples in $S_y$.}
    
    In $V_{y,0}$, the first $i-1$ rankings are $s_1, \ldots, s_{i-1}$ candidates. So, the next ranked candidate receives $p$ points for that vote, and the following receives $q$. In $n_y/2-1$ votes, $a$ is ranked ahead of $b$, and in $n_y/2+1$ votes, $b$ is ranked ahead of $a$ (or vice versa). Thus $a$ receives a partial score in $V_{y,0}$ of $p(n_y/2-1)+q(n_y/2+1)$ points, and $b$ receives $p(n_y/2+1)+q(n_y/2-1)$ points. Since $p>q$, $b$ has more points than $a$, so both cannot be winners.
    
    \noindent\textbf{Case~B: \boldmath$a$ is an element of a triple in $S_y$, but $b$ is not, or vice versa.}
    
    In $V_{y,0}$, $a$ and $b$ must be in the $i,i+1$th positions and receive $p$ or $q$ points respectively for each vote. Since $b\in A'$, $a$ will receive $p(n_y)$ points and $b$ will receive $q(n_y)$ points in $V_{y,0}$. Therefore, $a$ and $b$ are not both winners.
    
    \noindent\textbf{Case~C: \boldmath Neither $a$ nor $b$ are elements of a triple in $S_y$. }
    
    Suppose $a$ comes before $b$ in the lexicographic order. Then $a$ will receive $p$ points for each vote in $V_{y,0}$, and $b$ will receive $q$ for each vote in $V_{y,0}$. In other words, $a$ will receive $p(n_y)$ points and $b$ will receive $q(n_y)$ points in $V_{y,0}$, so both cannot be winners.
\end{proof}
We have proven the analogous result to Lemma~\ref{lemma:valid-cands}, but for any nontrivial scoring vector rather than just 1-approval. The rest of the proof follows from the end of the proof for 1-approval in Theorem~\ref{theorem:1-approval-l2}. Thus, $\calr\dd\$\crc{2}$ and $\calr\dd\crc{2}$ are \NP-complete for any nontrivial, pure scoring rule $\calr$.
\end{proof}

Extending our proof of Theorem~\ref{theorem:nontrivial-2} to hold for all $\ell \geq 3$ does not seem straightforward, as in our proofs some assignments of candidates are not successful because they elect more than 2 candidates in a district, which for larger values of $\ell$ may be completely acceptable.
Nonetheless for the remaining types of recampaigning not yet covered in this section, we show that for any voting rule in the $t$-approval and $t$-veto families of (nontrivial) scoring rules, bounded (with $\ell > 2$) and unbounded recampaigning is $\NP$-complete.

\begin{theorem}\label{theorem:3-and-above}
    For each 
    voting rule $\cale$ within the $t$-approval or $t$-veto families of voting rules 
    and each $\ell \geq 3$,
      $\cale\dd\crc{\ell}$, $\cale\dd\$\crc{\ell}$, $\cale\dd\crc{}$, and $\cale\dd\$\crc{}$ are all $\NP$-complete.
\end{theorem}
\begin{proof}
    Since all the above problems are in $\NP$, it suffices to prove $\NP$-hardness for the unpriced problems. Let $\ell = 3$ and let $t \geq 1$. Consider an X3C instance with $\calu=\{u_1,\ldots,u_{3m}\}$ and $\cals = \{S_1,\ldots,S_k\}$. Without loss of generality, assume $m > 1$ and $k > 1$.

    \noindent{\bfseries Proof for \boldmath$t$-approval.}
    In our construction, we define $k$ districts, and for each district we will give seven sets of buffer candidates as part of our construction.
    Let $A = \calu$. 
    For each $i \in [k]$ and for each $j \in [7]$, let
    $B^i_j = \{b^i_{j, 1}, \ldots, b^i_{j, t-1}\}$, which will denote each of the seven sets of buffer candidates for district $i$. We require that in each district, the sets of buffer candidates be pairwise disjoint, and that none of the buffer candidates be in $A$\@.
    For brevity, we write $B^i = B^i_1 \cup \cdots \cup B^i_7$.

    Let $A'_i = A - S_i$, where $S_i = \{s_{i, 1}, s_{i, 2}, s_{i, 3}\}$, and 
    for each $i \in [k]$, define district $D_i = (C_i, V_i)$, where
    $C_i = \{p_i\} \cup B^i$ such that $p_i$ is a new candidate not in $A \cup B^i$,
    and $V_i$ contains the following seven votes (over $C_i \cup A$):
    \begin{itemize}
        \item $B^i_1 > s_{i,1} > p_i > s_{i,2} > s_{i,3} > (B^i - B^i_1) \cup A'_i$
        \item $B^i_2 > s_{i,1} > p_i > s_{i,2} > s_{i,3} > (B^i - B^i_2) \cup A'_i$
        \item $B^i_3 > s_{i,2} > p_i > s_{i,3} > s_{i,1} > (B^i - B^i_3) \cup A'_i$
        \item $B^i_4 > s_{i,2} > p_i > s_{i,3} > s_{i,1} > (B^i - B^i_4) \cup A'_i$
        \item $B^i_5 > s_{i,3} > p_i > s_{i,1} > s_{i,2} > (B^i - B^i_5) \cup A'_i$
        \item $B^i_6 > s_{i,3} > p_i > s_{i,1} > s_{i,2} > (B^i - B^i_6) \cup A'_i$
        \item $B^i_7 > p_i > s_{i,1} > s_{i,2} > s_{i,3} > (B^i - B^i_7) \cup A'_i$
    \end{itemize}

    Notice that initially, in election $(C_i, V_i)$, each buffer candidate has 1 point and $p_i$ has 7 points, therefore making $p_i$ a unique winner of that election.
    
    \begin{lemma}\label{lemma:k-approval-winners}
        For each $i \in [k]$ and
        for each nonempty $\hat{A} \subseteq A$, the winner set of $(C_i \cup \hat{A}, V_i)$ under $t$-approval is $\hat{A}$ if and only if $\hat{A} = S_i$.
    \end{lemma}
    \begin{proof}
        Let $i \in [k]$ and let $\hat{A} \subseteq A$ be nonempty.
        
        Suppose the winner set of $(C_i \cup \hat{A}, V_i)$ under $t$-approval is $\hat{A}$.
        All candidates in $C_i - \{p_i\}$ have score 1, and $p_i$ has at least 1 point.
        Suppose $\hat{A} \cap A_i' \neq \emptyset$, then any such candidate will have 0 points, and thus lose, so $\hat{A} \subseteq S_i$.
        Moreover, $p_i$ is not a winner as $p_i$ cannot be in $\hat{A}$.
        Since each candidate in $S_i$ has at most 2 points, it must be that $p_i$ has at most 1 point. This is only possible when $\hat{A} = S_i$.
        Conversely, let $\hat A = S_i$.
        Then, $s_{i,1}$, $s_{i,2}$, and $s_{i,3}$ all have a score of 2, and all remaining candidates receive 1 point, so the winner set of $(C_i \cup \hat{A}, V_i)$ under $t$-approval is $\hat{A}$.
    \end{proof}

    We now give the correctness of our reduction by proving that $(\calu, \cals)$ is a Yes-instance of X3C if and only if $(A, D_1, \ldots, D_k)$ is a Yes-instance of $t\dd {\rm approval}\dd\crc{\ell}$.
    
    $\implies$:
    Suppose $(\calu, \cals)$ is a Yes-instance of X3C with certificate $\cals'$. Each element of $\cals'$ corresponds to a district where we can add the three candidates (elements in the triple), all of whom win by Lemma~\ref{lemma:k-approval-winners}. Since all elements in $\calu$ are in exactly one triple of $\cals'$, all elements of $A$ are added to and win in exactly one district. Therefore, this corresponds to a Yes-instance of $t\dd {\rm approval}\dd\crc{\ell}$.

    $\impliedby$:
    Suppose $(A, D_1, \ldots, D_k)$ is a Yes-instance of $t\dd {\rm approval}\dd\crc{\ell}$. Then for each district, we have either added 3 candidates or 0 candidates. The 3 added candidates correspond to a set containing them in $\cals$ by Lemma~\ref{lemma:k-approval-winners}, so the set of districts with added candidates corresponds to some $\cals'$ which is a certificate of X3C. All candidates in $A$ have been added to exactly one district, so all elements in $\calu$ are in exactly one triple of $\cals'$.

    By Lemma~\ref{lemma:k-approval-winners}, either zero or three candidates must be added to each district, regardless of $\ell$. Thus, for any $\ell\ge 3$, X3C polynomial-time many-one reduces to both $t\dd {\rm approval}\dd\crc{\ell}$ and $t\dd {\rm approval}\dd\crc{}$.

    \noindent{\bfseries Proof for \boldmath$t$-veto.}

    In this part of the proof, we only create one set of buffer candidates for each district.
    Let $A = \calu$.
    For each $i \in [k]$, let
    $B^i= \{b^i_{1}, \ldots, b^i_{t-1}\}$ such that $B_i \cap A = \emptyset$.
    Let $A'_i = A - S_i$, where $S_i = \{s_{i, 1}, s_{i, 2}, s_{i, 3}\}$, and 
    for each $i \in [k]$, define district $D_i = (C_i, V_i)$, where
    $C_i = \{p_i\} \cup B^i$ such that $p_i$ is a new candidate not in $A \cup B^i$,
    and $V_i$ contains the following five votes (over $C_i \cup A$):
     \begin{itemize}
        \item $p_i > s_{i,1} > s_{i,2} > s_{i,3} > \vec{A_i'} > B^i$
        \item $p_i > s_{i,2} > s_{i,3} > s_{i,1} > \vec{A_i'} > B^i$
        \item $p_i > s_{i,3} > s_{i,1} > s_{i,2} > \vec{A_i'} > B^i$
        \item $s_{i,1} > s_{i,2} > s_{i,3} > p_i > \vec{A_i'} > B^i$
        \item $s_{i,1} > s_{i,2} > s_{i,3} > p_i > \vec{A_i'} > B^i$
    \end{itemize}

    \begin{lemma}\label{lemma:k-veto-winners}
        For each $i \in [k]$ and
        for each nonempty $\hat{A} \subseteq A$, the winner set of $(C_i \cup \hat{A}, V_i)$ under $t$-veto is $\hat{A}$ if and only if $\hat{A} = S_i$.
    \end{lemma}
    \begin{proof}
        Let $i \in [k]$ and $\hat{A} \subseteq A$ be nonempty.
        
        Suppose $\hat{A} = S_i$.
        Therefore, in the $t$-veto election $(C_i \cup \hat{A}, V_i)$, each candidate in $\hat A$ gets 4 points, while $p_i$ gets 3 points, and the remaining candidates get 0 points, so the winner set is $\hat A$.

        Suppose the winner set of $(C_i \cup \hat{A}, V_i)$ under $t$-veto is $\hat{A}$. If $\hat{A} \cap A'_i \neq \emptyset$, then at least one such candidate (in the intersection) receives 0 points and loses and cannot win as $p_i$ receives at least 1 point. So $\hat{A} \subseteq S_i$. 
        If $\card{\hat A} = 1$, then the added candidate receives 2 points and loses to $p_i$ who receives 3 points. Thus $\card{\hat A} > 1$.
        If $\card{\hat A} = 2$, then only one of the two added candidates wins. Thus $\card{\hat A} > 2$.
        If $\card{\hat A} = 3$, then the winner set is $\hat A$. So $\hat{A} = S_i$.
    \end{proof}
    The remainder of the proof follows in the same way as the proof of the $t$-approval part.
\end{proof}

Though the results highlighted so far may paint a bleak picture for recampaigning as being inherently hard (computationally), 
there is also a significant class of natural voting rules, such as resolute voting rules like Condorcet and Ranked Pairs (see~\cite{rot:b:econ-second-edition} for definitions), under which those problems are in $\P$. We prove that as our next result.

\begin{theorem}\label{t:resolute-poly}
    Priced/unpriced bounded/unbounded recampaigning in is $\P$ under all polynomial-time computable resolute (i.e., $1\textsc{-Resolute}$) voting rules.
\end{theorem}
\begin{proof}
    Let $\cale$ be a polynomial-time computable resolute, i.e., 1-\textsc{Resolute}, voting rule. By Proposition~\ref{prop:k-resolute}, for each $\ell \geq 1$, $\cale\dd\crc{\ell} = \cale\dd\crc{}$ and also $\cale\dd\$\crc{\ell} = \cale\dd\$\crc{}$.
    Since $\cale$ is both 1-\textsc{Resolute} and polynomial-time computable, $\cale$ satisfies the 1-\textsc{Winner} condition, and thus both $\cale\dd\crc{1}$ and $\cale\dd\$\crc{1}$ are in $\P$, which concludes the proof.
\end{proof}

Moreover, we believe that Theorem~\ref{t:resolute-poly} highlights an important decision in the definition of voting rules. Pure (classical) social choice theory often requires a voting rule to elect a unique winner---and doing so at times involves defining a tie-breaking method---while it is common  for \emph{computational} social choice to allow voting rules to elect more than one winner, as we do in this paper. From a complexity perspective, such a seemingly innocuous decision can have drastic consequences in the recampaigning setting. For a detailed discussion on the naturalness of the empty winner set, we refer readers to~\cite[Footnote~3]{hem-hem-men:j:search-versus-decision}.

\section{Parametrizing by Number of Districts}\label{s:parametrized}

 Since recampaigning is NP-complete for a large class of voting rules that we care about, it is useful to consider also the cases where the problem may be tractable. 
 In particular, we focus on parametrizing with respect to the number of districts as that number may be much smaller than the size of the input, for example, if one is dealing with two districts, each with a large number of voters (or even candidates), it may not be immediately clear how the complexity of recampaigning is affected by this parametrization.

We find that in the bounded-recampaigning case, the problem is indeed fixed-parameter tractable under polynomial-time computable voting rules. We do so by drawing a relationship between the bound on the number of winners and on the number of districts, which allows us to determine in polynomial time No-instances where the number of additional candidates exceeds a given constant. This thus allows us to reduce the problem instance to a constant-sized (with respect to the parameter) instance of a related set covering problem, which we can solve efficiently (due to the new instance being constant-sized). We will first describe the just-mentioned reduction, and prove some helpful properties about it, before proving this subsection's main theorem.

\begin{construction}\label{construction:priced-crc-to-set-cover}
    Let $\cale$ be a voting rule, let $\ell$ be a positive integer, and let $I = (D_1, \ldots, D_k, A, \pi, B)$ an instance of $\cale\dd\$\crc{\ell}$. Assume without loss of generality that $A \cap [k] = \emptyset$.
    Define the sets $\calu$ and $S$, and the weight function $w$ as follows.
    \begin{itemize}
    \item $\calu = A \cup [k]$. 
    \item For each $A' \in 2^A$ and each $i \in [k]$, let $w(A'\cup \{i\}) = \sum_{a \in A'}\pi(i, a)$.
    \item For each $A' \in 2^A$ and each $i \in [k]$, $A' \cup \{i\} \in S$ if and only if 
    \begin{itemize}
        \item $\card{A'} = 0$ or
        \item $0 < \card{A'} \leq \ell$, $A' = \cale(C_i \cup A', V_i)$, and $w(A' \cup \{i\}) \leq B$.
    \end{itemize}
\end{itemize}
\end{construction}

In the above construction, finding a successful assignment of candidates that satisfies the budget $B$ corresponds exactly to finding an exact cover of $\calu$ in $S$ with weight at most $B$. , i.e., finding a set $S' \subseteq S$ of pairwise-disjoint sets that contain all the elements of $\calu$ such that $\sum_{T \in S'}w(T) \leq B$. 

\begin{lemma}\label{lemma:crc-set-cover-correctness}
    In Construction~\ref{construction:priced-crc-to-set-cover}, $I$ is a Yes-instance of $\cale\dd\$\crc{\ell}$ if and only if $S$ contains an exact cover of $\calu$ with weight at most $B$.
\end{lemma}
\begin{proof}

    $\implies:$ Suppose $I$ is a Yes-instance of $\cale\dd\$\crc{\ell}$ and let $A_1, \ldots, A_k$ be a partition of $A$ corresponding to the assignment of candidates in $A$ to districts. The cost of this assignment is $\sum_{i=1}^k \sum_{a \in A_i} \pi(i, a) \leq B$.
   
    Let $i \in [k]$ and let $S_i = A_i \cup \{i\}$.
    It follows that $S_i \in S$ and $w(S_i) = \sum_{a \in A_i} \pi(i, a)$.
    Since $A_1, \ldots, A_k$ is a partition of $A$, it follows that $S_1, \ldots, S_k$ is an exact cover of $\calu$. Moreover it has weight
    \begin{align*}
        \sum_{i=1}^{k} w(S_i) = \sum_{i=1}^k \sum_{a \in A_i} \pi(i, a) \leq B.
    \end{align*}
    
    Thus $S$ contains an exact cover of $\calu$ with weight at most $B$.

    $\impliedby$: Suppose $S$ contains an exact cover of $\calu$ with weight at most $B$. Since each $i \in [k]$ can only be in one element of the cover and all such elements must be in the cover, it follows that the cover has size $k$. Let $S_i$ be the element of the cover containing $i$ and let $A_i = S_i - \{i\}$. 
    
    Because $S_1, \ldots, S_k$ is an exact cover of $\calu$, every $a \in A$ is in some $A_i$ and no two sets in the cover intersect. We can thus conclude that $A_1, \ldots, A_k$ is a partition of $A$, and thus corresponds to an assignment of candidates to districts.
    Let $S_i \in S$. It follows that either $\card{A_i} = 0$, or $0 < \card{A_i} \leq \ell$ and $A_i = \cale(C_i \cup A_i, V_i)$, and so $A_1, \ldots, A_k$ corresponds to an assignment of candidates to districts such that every additional candidate is a winner in the district they win and there are at most $\ell$ winners in such districts. 
    
    It thus remains to show that the cost of the assignment is at most $B$. The cost of the assignment is
    \begin{align*}
        \sum_{i=1}^{k} \sum_{a \in A_i} \pi(i, a) = \sum_{i=1}^k w(A_i \cup \{i\}) = \sum_{i=1}^k w(S_i) \leq B.
    \end{align*}
\end{proof}

\begin{theorem}\label{theorem:fpt-districts}
    For each $\ell \geq 1$ and each voting rule $\cale$ satisfying $\ell$-{\rm \sc Winner}, $\cale\dd\$\crc{\ell}$ and $\cale\dd\crc{\ell}$ are fixed-parameter tractable with respect to the number of districts.
\end{theorem}
\begin{proof}
If $\cale$ satisfies $1$-\textsc{Winner}, then $\cale\dd\$\crc{1}$ and $\cale\dd\crc{1}$ are in $\P$, and thus fixed-parameter tractable.
Let $\ell \geq 2$ and let $\cale$ be a voting rule satisfying $\ell$-{\rm \sc Winner}. By Observation~\ref{observation:trivial-relationships}, it suffices to show that $\cale\dd\$\crc{\ell}$ is fixed-parameter tractable with respect to the number of districts.

Let $I = (D_1, \ldots, D_k, A, \pi, B)$ be an instance of $\cale\dd\$\crc{\ell}$ input to our $\FPT$ algorithm. The key observation for this result is the following: if $I$ is a Yes-instance, then a successful assignment of the $n$ additional candidates assigns at most $\ell$ candidates to each of the $k$ districts. Thus if $n > k\ell$, then $I$ cannot be a Yes-instance, so our algorithm rejects $I$ if $n > k\ell$.

If $n \leq k\ell$, since $\ell$ is a constant, then $n$ is also constant-bounded (with respect to our parameter $k$). Construct $\calu$, $S$ and $w$ per Construction~\ref{construction:priced-crc-to-set-cover}. If $S$ contains an exact cover of $\calu$ with weight at most $B$, our algorithm accepts. Otherwise it rejects.

By Lemma~\ref{lemma:crc-set-cover-correctness}, our algorithm is correct. It remains now to show that there is function $f$ such that our algorithm runs in time $f(k)|I|^{O(1)}$.

First, checking if $n \leq k \ell$ can be done in time polynomial in $|I|$.
Additionally,
$\card{\calu} = n+k \leq 2k+\ell$, $\card{S} \leq k2^n \leq k2^{k\ell}$, and the weight of each element of $S$ is at most $B$. 
Moreover,
one can see that
the time to perform the construction defined by Construction~\ref{construction:priced-crc-to-set-cover} is polynomial in $\card{\calu}$, $\card{S}$, and $|I|$.
Finally, to check if there is an exact cover of $\calu$ with weight at most $B$, it suffices to check every $k$-sized subset of $S$ to ensure they form an appropriate cover. There are $\binom{\card{S}}{k} = O(\card{S}^k)$ such sets, and verifying one such subset  is an exact cover with weight at most $B$ can be done in time polynomial in $\card{S}$ and $\card{B}$. 
\end{proof}

Next we demonstrate a gap between the complexity of bounded and unbounded recampaigning by showing that
unbounded recampaigning is para-\NP-hard with respect to the number of districts under the $t$-approval family of voting rules.\footnote{To prove that a problem is para-\NP-hard with respect to some parameter, it suffices to show that the problem is \NP-hard even when that parameter is a constant. Moreover, $\FPT = {\rm para}\dd\NP$-hard if and only $\P = \NP$. See~\cite{flu-gro:b:parameterized-complexity} for more details.}

\begin{theorem}\label{theorem:para-np-hard}
    For each $t \geq 1$, 
    unbounded recampaigning is {\rm para-NP}-hard, in both the priced and unpriced settings, under $t$-approval.
\end{theorem}
\begin{proof}

    Let $t$ be an arbitrary positive integer. We give the proof by giving a reduction from a variant of Monotone 1-in-3 SAT where each variable occurs exactly three times, which is NP-complete~\cite{moo-rob:j:monotone-1-3-sat}.

    \LANGDEF{Monotone Exact-3 1-in-3 SAT (E3-1/3-SAT$^+$)~\cite{moo-rob:j:monotone-1-3-sat}}
    {A 3-CNF boolean formula $f$ with only positive literals such that every variable occurs in exactly three clauses.}
    {Is there a satisfying assignment of $f$ to that exactly one literal in each clause is set to true?}

    To prove the para-\NP-hardness of interest, it suffices to prove that $t$-approval-CRC$_{\naturals}$ is \NP-hard for a constant number of districts (and by the transitivity of polynomial-time many-one reductions, the analogous result in the priced setting will follow).
    
    The intuition here is to make the additional candidates correspond to $f$'s variables and to partition the variables into two groups: group~1 contains the variables set to true, while group~2 contains the variables set to false. In the reduction, the groups are represented by districts, and the construction is such that each clause will have exactly one variable in group~1 and exactly two variables in group~2.

    Let $f$ be an instance of E3-1/3-SAT$^+$ with clauses $S_1, \ldots, S_m$ over variables $x_1, \ldots, x_n$. Assume without loss of generality that $m > 3$ and $n > 1$, as otherwise the problem can be solved in polynomial time. For simplicity, we interpret each clause as a set of variables.
    We construct two districts $D_1$ and $D_2$ as follows. 

    We first prove the case for 1-approval, and the construction generalizes to $t$-approval. For brevity, let $\calr = 1$-approval.
    Let the set of additional candidates be $A = \{x_1, \ldots, x_n\}$.
    Let $D_1 = (C_1, V_1)$, where $C_1 = \{s_1, \ldots, s_m\}$ and $V_1$ contains the following votes for each $i \in [m]$:
    \begin{itemize}
        \item one vote of the form $\overset{\rightarrow}{S_i} > s_i > \cdots$ and
        \item three votes of the form $s_i > \cdots$.
    \end{itemize}

    Furthermore, let $D_2 = (C_2, V_2)$, where $C_2 = \{t_1, \ldots, t_m\}$ and $V_2$ contains the following votes for each $i \in [m]$:
    \begin{itemize}
        \item one vote of the form $\overset{\rightarrow}{S_i} >  t_i > \cdots$,
        \item one vote of the form $\overset{\leftarrow}{S_i} >  t_i > \cdots$, and
        \item three votes of the form $t_i > \cdots$.
    \end{itemize}

    The reduction is polynomial-time computable, so we now show that $f$ has a satisfying assignment with exactly one literal in each clause is set to true if and only if $(A, D_1, D_2) \in \calr\dd\crc{}$.

    $\implies$: Suppose $f$ is in E3-1/3-SAT$^+$ and let $X$ denote the variables set to true such that each clause in $f$ has exactly one variable set to true. Consider the $\calr$ election $(C_1 \cup X, V_1)$. Then each $x \in X$ receives 3 points and each $s_i \in C_1$ receives 3 points. Thus the winner set contains $X$. 
    Moreover in the $\calr$ election $(C_2 \cup (A-X), V_2)$, each $a \in A-X$ receives 3 points and each $t_i \in C_2$ receives 3 points. This is because for a pair of votes of the form $\overset{\rightarrow}{S_i} >  t_i > \cdots$ and  $\overset{\leftarrow}{S_i} >  t_i > \cdots$, two candidates from $S_i$ must be present in $A-X$ and each of those votes gives a point to a different candidate. Thus $(A, D_1, D_2) \in \calr\dd\crc{}$

    $\impliedby$: Suppose $(A, D_1, D_2) \in \calr\dd\crc{}$. 
    Let $X$ denote candidates assigned to $D_1$ in a solution.
    Each vote of the form $S_i > s_i > \cdots$ in the first district can give a point to at most one additional candidate and each additional candidate must receive at least three points to win (since each $s_i$ will receive at least three points). This corresponds to at least one variable in each clause to being set to true. Suppose there are two distinct candidates $x_i, x_j \in X$ such that they both appear in the same clause $S_i$. Then one of them will receive at most two points and thus not win in $D_1$, which contradicts the assumption that $(A, D_1, D_2) \in \calr\dd\crc{}$. Thus there is a selection of variables that can be set to true so that exactly one variable in each clause is true.
    Thus $f$ is in E3-1/3-SAT$^+$.

    We now describe how to extend the above construction for $t$-approval when $t > 1$. Let $t \geq 2$ be an arbitrary integer.
    Let the set of additional candidates be $A = \{x_{1, i}, \ldots, x_{n, i} \mid i \in [t]\}$.
    Then in the above construction of $D_1$ and $D_2$, for each $i \in [n]$, replace each occurrence of  $x_i$ in a vote with $x_{i, 1}, \ldots, x_{i, t}$ (in that order), and for each $i \in [m]$, respectively replace each occurrence of $s_i$ and $t_i$ with $s_{i, 1}, \ldots, s_{i, t}$ and $t_{i, 1}, \ldots, t_{i, t}$ (in that order). In doing so, we have effectively replaced each candidate by $t$ copies of itself, so that a candidate receives $N$ points in the 1-approval construction if and only if its corresponding $t$ copies receive $N$ points in the new construction. The correctness argument is effectively the same as the one given for the 1-approval case.
\end{proof}

\section{Conclusion}

We have introduced a new model to capture the aftereffects of redistricting from the standpoint of electoral control. We have compared our newly introduced problems under the notions of separations and collapses, interreducibilities, and both worst-case and parametrized computational complexity. 
We find that the complexity of recampaigning can depend on a variety of considerations, such as properties of the underlying voting rule, the number of ``allowed winners,'' and whether the cost of assigning candidates to districts can be different.
We summarize the complexity results in Table~\ref{table:summary}.
We also find that bounded redistricting can be easier than winner evaluation, while unbounded redistricting is never easier than winner evaluation.
There are many interesting future directions here, and we list a few below. 

\begin{table}[ht]
    \centering
    \caption{Summary of computational complexities established.}
    \footnotesize
    \label{table:summary}
    \begin{tabular}{c | c c c || c}
       {Recampaigning} & \multicolumn{3}{c||}{Bounded} & \multirow{2}{*}{Unbounded}\\\cline{2-4}
       {Type} & $\ell = 1$ & $\ell = 2$ & $\ell \geq 3$ & {} \\ \hline
       Worst-Case & \makecell{$\P$ iff $\cale$ satisfies\\ 1-\textsc{Winner}} & \makecell{\NP-complete for all\\ pure scoring rules} & \makecell{\NP-complete for\\ $t$-approval and $t$-veto} & \makecell{\NP-complete for\\ $t$-approval and $t$-veto}\\\hline
       \makecell{Parametrizing\\by Number of\\Districts} & \multicolumn{3}{c||}{\FPT\ if $\cale$ satisfies $\ell$-\textsc{Winner}} & \makecell{para-\NP-hard\\for $t$-approval}
    \end{tabular}
\end{table}

\paragraph{Modeling.}
A natural next step would be to consider destructive recampaigning, and explore how it relates (if at all) to constructive recampaigning. We suggest below a definition for bounded destructive recampaigning for each voting rule $\cale$ and $\ell \geq 1$, and leave open its study for future work.

\LANGDEF{$\cale$ Destructive $\ell$-Recampaigning ($\cale\dash{\rm DRC}_{\leq \ell}$)}
{Districts $D_1, \ldots, D_k$, where for each $i \in [k]$, $D_i = (C_i, V_i)$ is an $\cale$ election, a set of $n$ additional candidates $A$ not present in any of the $k$ districts, and a collection $F \subseteq \bigcup_{i \in [k]} C_i$.}
{Is there an assignment of candidates in $A$ to districts such that each candidate $f \in F$ is either not a winner in its district or $f$'s district has more than $\ell$ winners.}

\paragraph{Relationships.}
It would be interesting to characterize the containment relationships between our variants. As we have highlighted, our work does not establish that there is some polynomial-time computable voting rule such that priced unbounded recampaigning does not reduce to unpriced bounded recampaigning. That case remains open. It would also be valuable to prove Proposition~\ref{prop:priced-unpriced} using only natural voting rules. To our knowledge, the only other result that demonstrates a complexity gap between priced and unpriced control also uses an unnatural voting rule~\cite{mia-fal:j:priced-control}.

\paragraph{Complexity.}
Are there natural voting rules and domain restrictions (e.g., single-peakedness) under which a recampaigning problem's complexity goes from $\NP$-hard to being in $\P$?
Do other types of parametrizations reveal interesting relationships?
We have studied parametrization by the number of districts, which revealed a difference between bounded and unbounded recampaigning. The other natural parametrization would be by the number of candidates. Notice that in both the bounded and unbounded cases, such recampaigning is in $\XP$ (for $k$ districts and $n$ additional candidates, the number of possible assignments of additional candidates to districts is bounded by~$k^n$), so there is room to improve that observation.
Another interesting direction would be to study approximation algorithms for recampaigning, perhaps focusing on maximizing the number of additional candidates that win or on maximizing the number of districts where an additional candidate wins (see~\cite{bui-cha-le-ngu:c:approximating-control} for the current state of approximations in control).

\bibliographystyle{alpha}

\begin{thebibliography}{CCH{\etalchar{+}}26}

\bibitem[BCLN26]{bui-cha-le-ngu:c:approximating-control}
H.~Bui, M.~Chavrimootoo, K.~Le, and S.~Nguyen.
\newblock Approximating electoral control problems.
\newblock In {\em Proceedings of the 9th International Conference on Algorithmic Decision Theory}, 2026.
\newblock To appear.

\bibitem[BD10]{bet-dor:j:possible-winner-dichotomy}
N.~Betzler and B.~Dorn.
\newblock Towards a dichotomy of finding possible winners in elections based on scoring rules.
\newblock {\em Journal of Computer and System Sciences}, 76(8):812--836, 2010.

\bibitem[Bed24]{bed:news:boebert}
J.~Bedayn.
\newblock Republican {US} {Rep}. {Lauren} {Boebert} wins after switching districts in {Colorado}.
\newblock {\em Associated Press News}, November 5, 2024.
\newblock Updated November 6, 2024.

\bibitem[BTT92]{bar-tov-tri:j:control}
J.~{{Bartholdi}}, III, C.~Tovey, and M.~Trick.
\newblock How hard is it to control an election?
\newblock {\em Mathematical and Computer Modeling}, 16(8--9):27--40, 1992.

\bibitem[CCF{\etalchar{+}}26]{cha-cli-fer-gib-hem-luu-nar-wan:c:axiomatic-tools}
M.~Chavrimootoo, I.~Clingerman, E.~Ferland, E.~Gibson, L.~Hemaspaandra, Q.~Luu, D.~Narv\'{a}ez, and Y.~Wang.
\newblock Axiomatic tools for separating electoral control types, with applications to concrete systems.
\newblock In {\em Proceedings of the 9th International Conference on Algorithmic Decision Theory}, 2026.
\newblock To appear.

\bibitem[CCH{\etalchar{+}}24]{car-cha-hem-nar-tal-wel:j:sct}
B.~Carleton, M.~Chavrimootoo, L.~Hemaspaandra, D.~Narv\'{a}ez, C.~Taliancich, and H.~Welles.
\newblock Separating and collapsing electoral control types.
\newblock {\em Journal of Artificial Intelligence Research}, 81:71--116, 2024.

\bibitem[CCH{\etalchar{+}}26]{car-cha-hem-nar-tal-wel:j:search-versus-search}
B.~Carleton, M.~Chavrimootoo, L.~Hemaspaandra, D.~Narv\'{a}ez, C.~Taliancich, and H.~Welles.
\newblock Search versus search for collapsing electoral control types.
\newblock {\em Theory of Computing Systems}, 70(44), 2026.

\bibitem[CD22]{cal-dan:news:australia-keneally}
K.~Calderwood and S.~Daniel.
\newblock Kristina {Keneally} concedes defeat in south-west seat of {Fowler}, previously held by {Labor} since 1984.
\newblock {\em Australian Broadcasting Corporation}, May 21, 2022.

\bibitem[Con85]{con:b:condorcet-paradox}
{J.-A.-N. de Caritat, Marquis de} Condorcet.
\newblock {\em Essai sur l'Application de L'Analyse \`{a} la Probabilit\'{e} des D\'{e}cisions Rendues \`{a} la Pluralit\'{e} des Voix}.
\newblock 1785.
\newblock Facsimile reprint of original published in Paris, 1972, by the Imprimerie Royale.

\bibitem[CZLR18]{coh-lew-ros:c:gerrymandering-graphs}
A.~Cohen-Zemach, Y.~Lewenberg, and J.~Rosenschein.
\newblock Gerrymandering over graphs.
\newblock In {\em Proceedings of the 17th International Conference on Autonomous Agents and Multiagent Systems}, pages 274--282. International Foundation for Autonomous Agents and Multiagent Systems, July 2018.

\bibitem[DW21]{duc-wal:b:political-geometry}
M.~Duchin and O.~Walch.
\newblock {\em Political Geometry}.
\newblock Springer, 2021.

\bibitem[EFPS20]{eib-fom-pan-sim:c:districts-fine-grained}
E.~Eiben, F.~Fomin, F.~Panolan, and K.~Simonov.
\newblock Manipulating districts to win elections: {Fine}-grained complexity.
\newblock In {\em Proceedings of the 34th AAAI Conference on Artificial Intelligence}, pages 1902--1909. AAAI Press, April 2020.

\bibitem[EHH15]{erd-hem-hem:c:more-natural-models-of-partition-control}
G.~Erd\'{e}lyi, E.~Hemaspaandra, and L.~Hemaspaandra.
\newblock More natural models of electoral control by partition.
\newblock In {\em Proceedings of the 4th International Conference on Algorithmic Decision Theory}, pages 396--413. Springer-Verlag {\it Lecture Notes in Artificial Intelligence \#9346}, September 2015.

\bibitem[FG06]{flu-gro:b:parameterized-complexity}
J.~Flum and M.~Grohe.
\newblock {\em Parameterized Complexity Theory}.
\newblock Springer-Verlag, 2006.

\bibitem[FGL{\etalchar{+}}16]{fal-gou-lan-les-mon:c:president-problem}
P.~Faliszewski, L.~Gourv\`{e}s, J.~Lang, J.~Lesca, and J.~Monnot.
\newblock How hard is it for a party to nominate an election winner?
\newblock In {\em Proceedings of the 25th International Joint Conference on Artificial Intelligence}, pages 257--263. AAAI Press, August 2016.

\bibitem[FH23a]{fit-hem:c:conformant-elections}
Z.~Fitzsimmons and E.~Hemaspaandra.
\newblock Complexity of conformant election manipulation.
\newblock In {\em Proceedings of the 24th International Symposium on Fundamentals of Computation Theory}, pages 176--189. Springer {Lecture Notes in Computer Science \#14292}, September 2023.

\bibitem[FH23b]{fit-hem:c:weighted-matching}
Z.~Fitzsimmons and E.~Hemaspaandra.
\newblock Using weighted matching to solve 2-{Approval}/{Veto} control and bribery.
\newblock In {\em Proceedings of the 26th European Conference on Artificial Intelligence}, pages 732--739. IOS Press, September 2023.

\bibitem[FR16]{fal-rot:b:handbook-comsoc-control-and-bribery}
P.~Faliszewski and J.~Rothe.
\newblock Control and bribery in voting.
\newblock In F.~Brandt, V.~Conitzer, U.~Endriss, J.~Lang, and A.~Procaccia, editors, {\em Handbook of Computational Social Choice}, pages 146--168. Cambridge University Press, 2016.

\bibitem[GJ79]{gar-joh:b:int}
M.~Garey and D.~Johnson.
\newblock {\em Computers and Intractability: {A} Guide to the Theory of {NP}-Completeness}.
\newblock {W. H. Freeman and Company}, 1979.

\bibitem[GJP{\etalchar{+}}21]{gup-jai-pan-roy-sau:c:gerrymandering-on-graphs-complexity}
S.~Gupta, P.~Jain, F.~Panolan, S.~Roy, and S.~Saurabh.
\newblock Gerrymandering on graphs: {Computational} complexity and parameterized algorithms.
\newblock In {\em Proceedings of the 14th International Symposium on Algorithmic Game Theory}, pages 140--155. Springer {Lecture Notes in Computer Science \#12885}, September 2021.

\bibitem[Hen24]{hen:news:voters-moving}
T.~Henderson.
\newblock Swing states see newcomers as {Americans} move from blue to red counties.
\newblock {\em Stateline}, April 3, 2024.

\bibitem[HHM20]{hem-hem-men:j:search-versus-decision}
E.~Hemaspaandra, L.~Hemaspaandra, and C.~Menton.
\newblock Search versus decision for election manipulation problems.
\newblock {\em ACM Transactions on Computation Theory}, 12(\#1, Article 3):1--42, 2020.

\bibitem[HHR07]{hem-hem-rot:j:destructive-control}
E.~Hemaspaandra, L.~Hemaspaandra, and J.~Rothe.
\newblock Anyone but him: {The} complexity of precluding an alternative.
\newblock {\em Artificial Intelligence}, 171(5--6):255--285, 2007.

\bibitem[HLM16]{hem-lav-men:j:schulze-and-ranked-pairs}
L.~Hemaspaandra, R.~Lavaee, and C.~Menton.
\newblock Schulze and ranked-pairs voting are fixed-parameter tractable to bribe, manipulate, and control.
\newblock {\em Annals of Mathematics and Artificial Intelligence}, 77(3--4):191--223, 2016.

\bibitem[Kar72]{kar:b:reducibilities}
R.~Karp.
\newblock Reducibility among combinatorial problems.
\newblock In R.~Miller and J.~Thatcher, editors, {\em Complexity of Computer Computations}, pages 85--103, 1972.

\bibitem[LL19]{lev-lew:c:reverse-gerrymandering}
O.~Lev and Y.~Lewenberg.
\newblock ``{Reverse} gerrymandering'': {Manipulation} in multi-group decision making.
\newblock In {\em Proceedings of the 33rd AAAI Conference on Artificial Intelligence}, pages 2069--2076. AAAI Press, July 2019.

\bibitem[LLR17]{lew-lev-ros:c:geographic-manipulation}
Y.~Lewenberg, O.~Lev, and J.~Rosenschein.
\newblock Divide and conquer: {Using} geographic manipulation to win district-based elections.
\newblock In {\em Proceedings of the 16th International Conference on Autonomous Agents and Multiagent Systems}, pages 624--632. International Foundation for Autonomous Agents and Multiagent Systems, May 2017.

\bibitem[MF15]{mia-fal:j:priced-control}
T.~Miasko and P.~Faliszewski.
\newblock The complexity of priced control in elections.
\newblock {\em Annals of Mathematics and Artificial Intelligence}, 77(3):225--250, 2015.

\bibitem[MR01]{moo-rob:j:monotone-1-3-sat}
C.~Moore and J.~Robson.
\newblock Hard tiling problems with simple tiles.
\newblock {\em Discrete \& Computational Geometry}, 26(4):573--590, 2001.

\bibitem[OCRS25]{okr-coh-rig-sch:news:tracking-redistricting}
A.~O'Kruk, E.~Cohen, R.~Rigdon, and F.~Schouten.
\newblock Tracking states' unprecedented redistricting efforts.
\newblock {\em CNN}, November 4, 2025.
\newblock Updated January 16, 2026.

\bibitem[Pal23]{pal:j:priced-gerrymandering}
D.~Palash.
\newblock Priced gerrymandering.
\newblock {\em Theoretical Computer Science}, 972(\#114080), 2023.

\bibitem[Rot24]{rot:b:econ-second-edition}
J.~Rothe, editor.
\newblock {\em Economics and Computation}.
\newblock Springer, 2nd edition, 2024.

\bibitem[RPSM24]{reg-plc-sai:news:farage}
H.~Regan, R.~Plcheta, and L.~Said-Moorhouse.
\newblock Farage wins first seat as his upstart right wing {Reform} {UK} party gains ground.
\newblock {\em CNN}, July 5, 2024.

\bibitem[RT12]{ram-tar:t:unbalanced-assignment}
L.~Ramshaw and R.~Tarjan.
\newblock On minimum-cost assignments in unbalanced bipartite graphs.
\newblock Technical Report HPL-2012-40R1, {HP} {Labs}, Palo Alto, CA, USA, 2012.

\bibitem[Tid87]{tid:j:clones-dodgson}
T.~Tideman.
\newblock Independence of clones as a criterion for voting rules.
\newblock {\em Social Choice and Welfare}, 4(3):185--206, 1987.

\bibitem[Unc24]{tim:news:gandhi}
Uncredited.
\newblock Rahul {Gandhi} to fight in {Wayanad}, no decision yet on {Amethi}.
\newblock {\em The Times of India}, March 8, 2024.

\bibitem[Wik26]{wik:url:mauritius-elections}
Wikipedia.
\newblock 2024 {Mauritian} general election.
\newblock \url{https://en.wikipedia.org/wiki/2024_Mauritian_general_election#Electoral_system}, 2026.

\end{thebibliography}
\newcommand{\etalchar}[1]{$^{#1}$}

\end{document}